\documentclass[12pt, draftclsnofoot, onecolumn,twoside]{IEEEtran}

\usepackage{url}
\usepackage{cite}
\usepackage{amsmath,amssymb,amsfonts}
\usepackage{amsthm}
\usepackage{algorithmic}
\usepackage{graphicx}
\usepackage{textcomp}
\usepackage{xcolor}
\usepackage{multicol}
\usepackage{booktabs}
\usepackage{epstopdf}
\usepackage{soul}
\usepackage{subfigure} 
\usepackage{lineno}
\newtheorem{theorem}{Theorem}
\newtheorem{lemma}{Lemma}

\newcommand{\snr}{\mathsf{snr}}
\newcommand{\sgn}{\text{sgn}}

\def\ve#1{{\mathchoice{\mbox{\boldmath$\displaystyle #1$}}%
		{\mbox{\boldmath$\textstyle #1$}}%
		{\mbox{\boldmath$\scriptstyle #1$}}%
		{\mbox{\boldmath$\scriptscriptstyle #1$}}}}

\def\e{{\mathrm{e}}}

\def\ptav{p_{\mathrm{T,av}}}

\renewcommand{\text}{\mathrm}

\begin{document}
	\title{Outage Performance Analysis of Widely Linear Receivers in Uplink Multi-user MIMO Systems}
	
	\author{\IEEEauthorblockN{Ronghua Gui, Naveen Mysore Balasubramanya, Lutz Lampe}}
	\markboth{Submitted paper}%
	{Submitted paper}
	
	\maketitle
	
	\begin{abstract}
		This paper considers the application of widely linear (WL) receivers in an uplink multi-user system using real-valued modulation schemes, where the cellular base station (BS) with multiple antennas provides connectivity for randomly deployed single-antenna users.  The targeted use case is massive machine type communication (mMTC) with grant-free access in the uplink, where the network is required to host a large number of low data rate devices transmitting in an uncoordinated fashion. Four types of WL receivers are investigated, namely the WL zero-forcing (ZF) and the WL minimum mean-squared error (MMSE) receivers, along with their enhanced versions employing successive interference cancellation (SIC) with channel-dependent ordering, i.e., the WL-ZF-SIC and WL-MMSE-SIC receivers. The outage performances of these receivers are analytically characterized in the high signal-to-noise ratio (SNR) regime and  compared to those of conventional linear (CL) receivers using complex-valued modulation schemes. For the non-SIC receivers, we show that, when compared to the CL counterparts, the WL receivers yield a higher diversity gain when decoding the same number of users and have the same diversity gain but a decreased coding gain when the number of users is nearly doubled. The outage performance analysis of WL-SIC receivers is facilitated by the marginal distribution of ordered eigenvalues of a real-valued Wishart matrix. It is shown that the SIC operation with channel-dependent ordering brings no additional diversity gain to the WL receivers but instead increases the coding gain. Moreover, the coding gain of WL-SIC receivers grows as the number of users increases and even exceeds that of CL-SIC receivers under suitable conditions. For the mMTC scenario with grant-free transmission, it is demonstrated that the WL receivers outperform their CL counterparts in terms of offering a lower outage (and packet drop) probability and a higher system throughput for a given packet drop probability.
	\end{abstract}

	\begin{IEEEkeywords}
		Multi-user communication, grant-free access, massive machine type communication (mMTC), multi-user detection, widely linear (WL) receiver, multi-input multi-output (MIMO), zero-forcing (ZF), minimum mean-squared error (MMSE), successive interference cancellation (SIC), outage probability, real Wishart matrix, ordered eigenvalues.
	\end{IEEEkeywords}
	\IEEEpeerreviewmaketitle

	\section{Introduction}
	Uplink multi-user communication to a base station (BS) with multiple receiver antennas forms a multi-input multi-output (MIMO) system, whose performance is dominated by the coordination and processing of multi-user signals. A key innovation has been the conception of multi-user detection techniques, which permit the simultaneous reliable transmission of multiple user streams \cite{Verdu:1998:MD:521411,TseViswanathFundamentals05}. Among various multi-user detection techniques, linear multi-user and successive interference cancellation (SIC) detection are of particular interest due to their favorable complexity-performance trade-offs. 
	
	The performance of conventional linear (CL) multi-user detection can, in certain scenarios, be enhanced by widely linear (WL) processing \cite{PicinbonoChevalierWidely95}. In the context of data communication, while CL processing refers to the use of complex-valued signal constellations for signal transmission and reception, WL processing refers to the use of one-dimensional constellations transmitted over a complex-valued channel. This specifically includes real-valued transmission with pulse amplitude modulation (PAM), but also complex-valued modulations such as offset quadrature amplitude modulation (offset QAM), minimum shift keying (MSK) and Gaussian minimum shift keying (GMSK), that can be interpreted as PAM transmission after a derotation operation at the receiver \cite{GerstackerSchoberReceivers03,ChevalierPiponNew06}. 
 As an important instance of this, WL processing has been adopted at the receiver to perform single antenna interference cancellation (SAIC) in cellular networks involving one-dimensional modulated signals, such as, binary phase shift keying (BPSK), MSK, and GMSK \cite{GerstackerSchoberReceivers03,ChevalierPiponNew06}. In this scenario, it has been demonstrated that the optimal WL receiver with an array of $N$ antennas can process up to $2N-1$ interference signals, while the conventional receiver requires $2N$ antennas to do the same \cite{ChevalierPiponNew06}.
	
	More generally, it has been shown that WL processing can improve the receiver performance in MIMO systems if the received signal is second-order noncircular or improper \cite{AdaliSchreierComplex11}, such as in the case of PAM or GMSK transmission. For example, the receiver performance improvement obtained using WL detection for single user MIMO systems is demonstrated in \cite{BuzziLopsWidely,MatteraPauraWidely,WitzkeLinear05,AghaeiPlataniotisWidely10}. Particularly, in \cite{WitzkeLinear05,AghaeiPlataniotisWidely10}, the transmitted codes are constructed using a linear combination of input symbols and their complex conjugate, which aids in WL processing. The applications of WL processing in various multi-user MIMO systems with real-valued constellations were investigated in \cite{KuchiPrabhuPerformance09,KuchiPrabhuInterference09,ChevalierDupuyWidely11,ZareiGerstackerLow15}. 

In the present age, numerous Internet of things (IoT) applications, such as, smart metering systems and smart buildings, are envisioned under the massive machine type communication (mMTC) framework, which hosts a large number of users demanding low data rates while being delay tolerant. In such scenarios, grant-free transmission is the preferred multiple access strategy \cite{AzariPopovskiGrant17, BockelmannPratasTowards18, cirik2019toward}.  Our previous works \cite{GuiBalasubramanyaUplink19,GuiBalasubramanyaConnectivity19} demonstrated that the user density in mMTC systems with grant-free access can be further enhanced by adopting real-valued modulation along with WL detection at the base station.  We presented an analytical outage performance characterization of the WL zero-forcing (WL-ZF) and WL minimum mean-squared error (WL-MMSE) receivers in these works. Specifically, in \cite{GuiBalasubramanyaConnectivity19}, we considered the packet drop probability for a fixed target rate for both CL and WL users and grant-free access in an mMTC scenario as a performance indicator. Based on this metric, we computed the supported user density as the maximum number of supported users per cell such that the packet drop probability does not exceed a given threshold (please see \cite[eq. (39)]{GuiBalasubramanyaConnectivity19}, and the system throughput as the number of correctly decoded packets per second per Hz (please see \cite[eq. (40)]{GuiBalasubramanyaConnectivity19}). Further, we demonstrated that the packet drop probability was dominated by the collision probability of the users. We showed that, for a fixed target data rate and transmission time interval (TTI) of the mMTC user application, (i) WL receivers are capable of resolving more colliding packets, thereby resulting in a lower packet drop probability than the CL receivers for a given number of users, or alternatively, (ii) WL processing supports a higher user density and offers a larger system throughput than CL processing for a given packet drop probability.

In this work, we further analyze the outage performance of WL receivers for uplink multi-user MIMO systems. While the analysis is agnostic with regards to specific applications, the targeted use-case is mMTC using grant-free access mechanisms as considered in our previous works \cite{GuiBalasubramanyaUplink19,GuiBalasubramanyaConnectivity19}. We derive the analytical expressions for diversity and coding gains of the WL-ZF and WL-MMSE receivers, based on their outage probability in the high signal-to-noise ratio (SNR) regime. These expressions enable us to analytically characterize the performance of these WL  receivers in high and low data rate scenarios, unlike the characterization based on heuristic arguments presented in our previous works \cite{GuiBalasubramanyaUplink19,GuiBalasubramanyaConnectivity19}. We also introduce the SIC counterparts of these WL receivers, i.e., WL-ZF-SIC and WL-MMSE-SIC receivers, analyze their outage performance and demonstrate the system throughput obtained by the use of the SIC-based WL receivers in an mMTC scenario with grant-free access. Moreover, the outage performance analysis of WL-SIC receivers is facilitated by the marginal distribution of the ordered eigenvalues of real-valued Wishart matrix, which is constructed from the complex channel matrix. However, for real-valued Wishart matrices, there are no closed-form expressions for the marginal distribution of their ordered eigenvalues, that can directly be employed or derived from the existing results for complex-valued Wishart matrices \cite{KrishnaiahChangOn71, EdelmanEigenvalues89,JohnstoneOn01,RatnarajahVaillancourtEigenvalues04,FeldheimSodinA08,
		ZhongMcKayDistribution11,ChianiDistribution14}. Therefore, we derive the polynomial approximation of the marginal cumulative distribution function (CDF) of those ordered eigenvalues around zero, which is shown to be accurate enough for the outage performance analysis of WL-SIC receivers in the high SNR regime. It is worth pointing out that our results of the polynomial approximation for the ordered-eigenvalue distribution of real Wishart matrix are also novel.
	

The remainder of this paper is organized as follows. Section~\ref{s:model} describes the system model for the uplink multi-user MIMO system with real-valued modulation schemes and presents the WL versions of conventional ZF and MMSE receivers. Section~\ref{s:performance} analyzes the outage performance of WL-ZF and WL-MMSE receivers. Section~\ref{s:performancesic} covers the outage performance analysis of the corresponding WL-SIC receivers. Numerical results are presented in Section~\ref{s:results}, followed by conclusions in Section~\ref{s:conclusions}.

	\section{System Model and Widely Linear Receivers}	
\label{s:model}

	\subsection{System Model}
	We consider a cellular uplink transmission scenario, in which the BS with $M$ receive antennas 
 supports a population of single-antenna users. Assuming that $N$ users are active in the current transmission interval, the complex baseband model for the multi-user uplink channel is given by 
	\begin{equation}\label{MimoMod}
	\ve{\bar{y}}=\sqrt{\ptav}\ve{\bar{H}}\ve{\varPsi}^{\frac{1}{2}}\ve{x}+\ve{\bar{n}},
	\end{equation}
	where $\ve{\bar{y}}\in \mathbb{C}^{M\times 1}$ is the vector of complex received samples, $\ptav$ denotes the average transmit power, and $\ve{\bar{H}}\in \mathbb{C}^{M\times N}$ is the matrix of $M\times N$ small-scale fading coefficients. For concreteness, we assume independent and circularly symmetric standard complex Gaussian distributed entries of $\ve{\bar{H}}$, i.e., $\bar{h}_{i,j}\sim \mathcal{C}\mathcal{N}\left( 0,1 \right)$ for $1\le i \le M$, $1 \le j \le N$. The diagonal matrix $\boldsymbol{\varPsi }\triangleq \text{diag}([\xi _1, \xi _2, \ldots , \xi _N])$ accounts for the received power variation due to power control and large-scale fading. Its $i^{\mathrm{th}}$ diagonal entry can be expressed as
	\begin{equation}
	\xi_i=\frac{p_{\text{T,}i}\beta _i\psi _i}{p_{\text{T,av}}},
	\end{equation}
	where $p_{\text{T,}i}$, $\beta _i$ and $\psi _i$ denote the instantaneous transmit power, the pathloss and the shadowing fading factor for user $i$, respectively. We assume that $\xi_i$, $1\le i\le N$, are independently and identically distributed (i.i.d.) random variables, whose distribution is determined by the user deployment, large-scale fading model and power control strategies \cite{HaenggiStochastic12}. Without power control, we have $p_{\text{T,}i}= p_{\text{T,av}}$, $1\le i\le N$, and the power variation term $\xi_i=\beta_i\psi_i$ is the large-scale fading. With perfect power compensation (PPC), $p_{\text{T,}i}=\ptav\frac{\left( \beta _i\psi _i \right) ^{-1}}{\mathbb{E}_{\beta_i \psi_i}\left\{ ( \beta _i\psi _i) ^{-1} \right\}}$, where $\mathbb{E}_{x}\{ \cdot \}$ denotes the statistical expectation with respect to $x$, and the large-scale fading is fully compensated so that $\xi_i=\frac{1}{\mathbb{E}_{\beta_i \psi_i}\left\{ ( \beta _i\psi _i) ^{-1} \right\}}\triangleq \xi_{\text{PPC}}$, $1\le i\le N$. Other cases lie in between those two extremes. The vector $\ve{x}$ contains the data symbols simultaneously transmitted by the $N$ users, which are generated from one-dimensional constellations, i.e., $\ve{x}\in \mathbb{R}^N$. As the $N$ data streams come from multiple independent users, we assume that $\mathbb{E}_{\ve{x}}\left\{ \ve{xx}^T \right\} =\ve{I}_N$, where $(\cdot)^T$ and $\ve{I}_N$ denote the transpose operator and the identity matrix of size $N \times N$, respectively. The complex noise $\ve{\bar{n}}\in \mathbb{C}^N$ is modelled to be white Gaussian with $\ve{\bar{n}}\sim \mathcal{C}\mathcal{N}\left( 0,\sigma _{\mathrm{c}}^{2} \right)$, where $\sigma _{\mathrm{c}}^{2}$ is the variance. For further use, we define the transmit SNR as
	\begin{equation}\label{e:snr}
	\snr \triangleq \frac{p_{\mathrm{T,av}}}{\sigma _{\mathrm{c}}^{2}}.
	\end{equation}
	In addition to typical real-valued modulations, i.e., PAM, the MIMO model considered in (\ref{MimoMod}) is also applicable to complex-valued modulation schemes, such as Gaussian minimum-shift keying (GMSK), whose complex amplitude can be considered as a filtered version of a real-valued modulation after a derotation operation \cite{ChevalierDelmasProperties14}.
	
	Considering the real-valued transmission, we apply the WL transform\footnote{The WL transform herein augments the $M$-receiver antenna into an equivalent $2M$-receiver antenna MIMO system, which allows us to apply the methodology for the analysis of conventional MIMO systems from e.g.\ \cite{JiangVaranasiPerformance11, TseViswanathFundamentals05}. It is different from the transform that stacks the complex-valued signal and its complex conjugate, which has widely been used in the literature \cite{PicinbonoChevalierWidely95,AdaliSchreierComplex11,BuzziLopsWidely,SchoberGerstackerData04,ChevalierPiponNew06,KorpiAnttilaWidely} for the sake of convenience in differentiating a real-valued function with respect to a complex argument.} $\mathcal{T}:\ \ve{\bar{x}}\in \mathbb{C}^{M}\,\rightarrow \,\ve{x}\in \mathbb{R}^{2M}$ with $\ve{x}=\left[ \text{Re}\left( \ve{\bar{x}}^T \right) \ \text{Im}\left( \ve{\bar{x}}^T \right) \right] ^T$, where $\text{Re}\left( \cdot \right)$ and $\text{Im}\left( \cdot \right)$ denote the real and imaginary parts, respectively, of the complex-valued vector $\ve{\bar{y}}$ defined above. Then we can rewrite the signal model in (\ref{MimoMod}) as
	\begin{equation}\label{WLM}
	\ve{y}=\sqrt{p_{\mathrm{T,av}}}\ve{H\varPsi }^{\frac{1}{2}}\ve{x}+\ve{n}
	\end{equation}
	where $\ve{H}=\left[ \text{Re}\left( \ve{\bar{H}}^T \right), \text{Im}\left( \ve{\bar{H}}^T \right) \right] ^T$ is the $2M \times N$ real-valued channel matrix with i.i.d. entries $h_{i,j}\sim \mathcal{N}\left( 0,0.5 \right)$ for $1\le i \le 2M$, and $\ve{n}=\left[ \text{Re}\left( \ve{\bar{n}}^T \right) ,\text{Im}\left( \ve{\bar{n}}^T \right) \right] ^T$ is the $2M \times 1$ real-valued noise vector with $\ve{n}\sim \mathcal{N}\left( 0,0.5\sigma _{\mathrm{c}}^{2}\ve{I}_N \right)$. The WL transform virtually extends the complex data vector $\ve{\bar{y}}\in \mathbb{C}^M$ to the real vector\footnote{The log-likelihood ratio (LLR) computation in WL detection will be same as that adopted in CL processing. That is, after linear or SIC processing, an effective additive Gaussian noise (AWGN) channel is assumed and LLRs are computed based on, for example, the max-log approximation.} $\ve{y}\in \mathbb{R}^{2M}$. 
	
	

	One could employ the CL-ZF and CL-MMSE receivers to the original complex-valued data vector $\ve{\bar{y}}$ in (\ref{MimoMod}) for decoding the $N$ user data streams, regardless of the prior knowledge about real-valued transmission. Consequently, the maximum number of detectable users for the CL receivers would be limited to be $M$. However, with the WL transform, the number of receive antennas will be virtually doubled and thus up-to-$2M$ user data streams can simultaneously be decoded by applying the WL-ZF or WL-MMSE detector to the real-valued data vector $\ve{y}$ in \eqref{WLM}, as briefly discussed next.

	\subsection{Widely Linear Receivers}
	We assume that the maximum number of detectable users is restricted to be $N\le 2M$, and the channel state information (CSI) can be perfectly tracked by the BS receiver, i.e., $\ve{H\varPsi}^{\frac{1}{2}}$ is known at the receiver. In analogy to the CL multi-user detection, we consider the ZF and MMSE criteria for the WL receiver model (\ref{WLM}).  One can obtain the detection matrix $\ve{W}\in \mathbb{R}^{2M\times N}$ using the expressions for CL multi-user detection \cite{TseViswanathFundamentals05} as 
	\begin{equation}\label{WLD}
	\ve{W}=\begin{cases}
	\ve{H\varPsi }^{\frac{1}{2}}\left( \ve{\varPsi }^{\frac{1}{2}}\ve{H}^T\ve{H\varPsi }^{\frac{1}{2}} \right) ^{-1}\,\,\,\,\text{for\,\,WL~ZF}\\
	\ve{H\varPsi }^{\frac{1}{2}}\left( \ve{\varPsi }^{\frac{1}{2}}\ve{H}^T\ve{H\varPsi }^{\frac{1}{2}}+\frac{1}{2\snr}\ve{I}_N \right) ^{-1}\,\,\,\,\text{for\,\,WL~MMSE}.\\
	\end{cases}
	\end{equation}
	Applying $\ve{W}$ to the real-valued data model \eqref{WLM}, we obtain the output signal-to-interference-plus-noise ratios (SINRs) of WL-ZF and WL-MMSE detectors for the $n^{\mathrm{th}}$ data stream as \cite{JiangVaranasiPerformance11,GuiBalasubramanyaConnectivity19}
	\begin{subequations}\label{ZFSINR}
		\begin{align}
		\gamma _{\text{WL-ZF},n}&=\frac{2\snr\xi _n}{\left[ \left( \ve{H}^T\ve{H} \right) ^{-1} \right] _{n,n}}
		\\
		&=2\snr\xi _n\ve{h}_{n}^{T}\ve{P}_{\ve{H}_n}^{\bot}\ve{h}_n
		\end{align}
	\end{subequations}
	and
	\begin{subequations}\label{MMSESINR}
		\begin{align}
		\gamma _{\text{WL-MMSE},n}&=\frac{2\snr\xi _n}{\left[ \left( \ve{H}^T\ve{H}+\frac{1}{2\snr}\ve{\varPsi }^{-1} \right) ^{-1} \right] _{n,n}}-1
		\\
		&=2\snr\xi _n\ve{h}_{n}^{T}\ve{\tilde{P}}_{\ve{H}_n}^{\bot}\ve{h}_n
		\end{align}
	\end{subequations}
	where $\ve{h}_n$ is the $n^{\mathrm{th}}$ column of $\ve{H}$ and
	\begin{subequations}
		\begin{align}
		\ve{P}_{\ve{H}_n}^{\bot}&=\ve{I}_{2M}-\ve{H}_n\left( \ve{H}_{n}^{T}\ve{H}_n \right) ^{-1}\ve{H}_{n}^{T}\label{PH1}
		\\
		\ve{\tilde{P}}_{\ve{H}_n}^{\bot}&=\ve{I}_{2M}-\ve{H}_n\left( \ve{H}_{n}^{T}\ve{H}_n+\frac{1}{2\snr}\ve{\varPsi }_{n}^{-1} \right) ^{-1}\ve{H}_{n}^{T}.\label{PH2}
		\end{align}
	\end{subequations}
	In $\ve{P}_{\ve{H}_n}^{\bot}$ and $\ve{\tilde{P}}_{\ve{H}_n}^{\bot}$, $\ve{H}_n\in \mathbb{R}^{2M\times \left( N-1 \right)}$ is $\ve{H}$ with the $n^{\mathrm{th}}$ column removed and $\ve{\varPsi }_n\in \mathbb{R}^{\left( N-1 \right) \times \left( N-1 \right)}$ is $\ve{\varPsi }$ with the $n^{\mathrm{th}}$ column and $n^{\mathrm{th}}$ row removed.
	
	Similar to their CL counterparts, the WL receivers discussed above enable relatively low-complexity detection at the price of sub-optimal error-rate performance. Considering that the performance of CL receivers can be enhanced using SIC with channel-dependent ordering \cite{TseViswanathFundamentals05}, we attempt to improve the performance of WL receivers using the same methodology. Specifically, for CL-SIC receivers, the optimal ordering is to choose the user to decode such that the output SINR is maximized at each decoding stage \cite{FoschiniGoldenSimplified99}. We apply this SINR-maximization ordering rule to the WL-SIC receivers, and then obtain the $n^{\mathrm{th}}$-stage output SINRs of the WL-ZF-SIC and WL-MMSE-SIC receivers as
	\begin{subequations}
	\begin{align}
	\gamma _{\text{WL}-\text{ZSIC}}^{\left( n \right)}&=\max _n\left\{ \gamma _{\text{WL}-\text{ZF,}1}^{\left( n \right)},\gamma _{\text{WL}-\text{ZF,}2}^{\left( n \right)},\cdots ,\gamma _{\text{WL}-\text{ZF,}N-n+1}^{\left( n \right)} \right\} 
	\\
	\gamma _{\text{WL}-\text{MSIC}}^{\left( n \right)}&=\max _n\left\{ \gamma _{\text{WL}-\text{MMSE,}1}^{\left( n \right)},\gamma _{\text{WL}-\text{MMSE,}2}^{\left( n \right)},\cdots ,\gamma _{\text{WL}-\text{MMSE,}N-n+1}^{\left( n \right)} \right\} 
	\end{align}
	\end{subequations}
where $\gamma _{\text{WL}-\text{ZF,}i}^{\left( n \right)}$ and $\gamma _{\text{WL}-\text{MMSE,}i}^{\left( n \right)}$ represent the $n^{\mathrm{th}}$-stage output SINRs of the WL-ZF and WL-MMSE detectors, respectively, for the $i^{\mathrm{th}}$ data stream that remains after previous $(n-1)$ SIC stages, i.e., $1\le i\le N-n+1$.
	
	\section{Outage Performance of Widely Linear Receivers}
\label{s:performance}

 In this section, we consider the  WL-ZF and WL-MMSE receivers using $\ve{W}$ from \eqref{WLD}. We first briefly state the SINR and outage probability results from our previous work \cite{GuiBalasubramanyaConnectivity19}, which are then used to establish the diversity and coding gains. 

The outage probability for WL receivers can be expressed as \cite{TseViswanathFundamentals05}
	\begin{equation}\label{Pout}
	\mathcal{P}= \mathrm{Pr}\left( \frac{1}{2}\log \left( 1+\text{SINR} \right) \le R \right)
	\end{equation}
	where the factor $1/2$ is due to the use of real-valued transmission and $R$ is the target data rate in bits/sec/Hz. Based on the behaviour of $\mathcal{P}$ in the high SNR regime, the outage probability of a communication system as defined in \eqref{Pout} can be expressed as \cite{TseViswanathFundamentals05}
	\begin{equation}
	\mathcal{P}\left( \snr\right) \simeq \left( C\cdot \snr \right) ^{-d}
	\end{equation}
	where $C$ is the outage coding gain, $d$ is the outage diversity gain, and $\simeq$ denotes the asymptotic equality in the high SNR regime. Using the diversity and coding gains, we can make an analytical comparison between the asymptotic outage behaviours of WL and CL receivers.

	\subsection{Previous Results\cite{GuiBalasubramanyaConnectivity19}}
	\subsubsection{SINR Distribution}The SINR of the WL-ZF receiver is distributed according to 
\begin{equation}\label{e:SINRZF}
	\frac{\gamma _{\text{WL}-\text{ZF},n}}{\xi _n}\sim \snr\cdot\chi _{2M-N+1}^{2}\;,
	\end{equation}
	where $\chi _{2M-N+1}^{2}$ represents the standard Chi-squared distribution with $2M-N+1$ degrees of freedom. As it is difficult to directly characterize the SINR distribution of the WL-MMSE receiver, we decompose the scaled SINR of WL-MMSE receiver, i.e., $\frac{\gamma _{\text{WL}-\text{MMSE,}n}}{\xi _n}$ as the sum of the scaled WL-ZF SINR, i.e., $\frac{\gamma _{\text{WL}-\text{ZF,}n}}{\xi _n}$ and the residual term $\eta _{\text{WL,}n}$ given by
	\begin{equation}\label{MMSE-SINR}
	\frac{\gamma _{\text{WL}-\text{MMSE,}n}}{\xi _n}=\frac{\gamma _{\text{WL}-\text{ZF,}n}}{\xi _n}+\eta _{\text{WL,}n}\;,
	\end{equation}
	where the residual term $\eta _{\text{WL},n}\triangleq \frac{\gamma _{\text{WL}-\text{MMSE},n}}{\xi _n}-\frac{\gamma _{\text{WL}-\text{ZF},n}}{\xi _n}=2\snr\boldsymbol{h}_{n}^{T}\left( \boldsymbol{P}_{\boldsymbol{H}_n}^{\bot}-\boldsymbol{\tilde{P}}_{\boldsymbol{H}_n}^{\bot} \right) \boldsymbol{h}_n$ can be shown to be independent of $\gamma _{\text{WL-ZF},n}/\xi _n$ and approximated as 
	\begin{equation}\label{SINRdiff2}
	\eta _{\text{WL},n}\simeq\ve{h}_{n}^{T}\ve{H}_n\left( \ve{H}_{n}^{T}\ve{H}_n \right) ^{-1}\ve{\varPsi }_{n}^{-1}\left( \ve{H}_{n}^{T}\ve{H}_n \right) ^{-1}\ve{H}_{n}^{T}\ve{h}_n
	\end{equation}
	in the high SNR regime, i.e., $\snr \rightarrow \infty$. Intuitively, the term $\eta _{\text{WL},n}$ represents the power of the signal component ``hidden'' in the range of $\boldsymbol{H}_n$ that is recovered by the
	WL-MMSE receiver, but nulled out by the WL-ZF receiver.

	In general, it is difficult to derive the analytical distribution of SINR gap $\eta _{\text{WL},n}$. However, for the case of perfect compensation of large-scale fading  we have $\xi _n=\xi _{\text{PPC}}$, $1 \le n\le N$, and thus $\ve{\varPsi }=\xi_{\text{PPC}}\ve{I}_N$. Then, we can provide the closed-form distribution for the scaled version of $\eta _{\text{WL},n}$ as
	\begin{equation}\label{e:addterm}
	\frac{2M-N+2}{N-1}\xi _{\text{PPC}}\cdot \eta _{\text{WL},n}\sim \mathcal{F}_{N-\text{1},2M-N+2},
	\end{equation}
where $\mathcal{F}_{N-\text{1},2M-N+2}$ denotes the standard $F$-distribution with degree-of-freedom (DoF) parameters $d_1=N-\text{1}$ and $d_2=2M-N+2$ as defined in \cite[Footnote 2]{JiangVaranasiPerformance11}. Note that the F distribution is interpreted as the ratio of two chi-squared distributions with degrees of freedom $d_1$ and $d_2$ respectively. 

    \subsubsection{Outage Probability} Applying the SINR distributions from above to the outage probability \eqref{Pout}, we obtain
    \begin{subequations}\label{e:outageprob}
    \begin{align}
		\begin{split}
		\mathcal{P}_{\text{WL}-\text{ZF}}&=\text{Pr}\left( \frac{1}{2}\log \left( 1+\gamma _{\text{WL}-\text{ZF},n} \right) \le R \right )\\
		&=\text{Pr}\left( \frac{\gamma _{\text{WL}-\text{ZF},n}}{\snr\cdot\xi _n}\le \frac{2^{2R}-1}{\snr\cdot\xi _n} \right )\\ &=\mathbb{E}_{\xi _n}\left\{ F_{\chi _{2M-N+1}^{2}}\left( \frac{\gamma _{\text{WLT}}}{\snr\cdot\xi _n} \right) \right\} \label{Pzf}
		\end{split}
    	\\
		\begin{split}
		\mathcal{P}_{\text{WL}-\text{MMSE}}&=\text{Pr}\left( \frac{1}{2}\log \left( 1+\gamma _{\text{WL}-\text{MMSE},n} \right) \le R \right)
		\\
		&\simeq\text{Pr}\left( \frac{\gamma _{\text{WL}-\text{ZF},n}}{\snr\cdot\xi _n}+\frac{\eta _{\text{WL},n}}{\snr}\le \frac{2^{2R}-1}{\snr\cdot\xi _n} \right)\\
		&=\mathbb{E}_{\xi _n,\eta _{\text{WL},n}}\left\{ F_{\chi _{2M-N+1}^{2}}\left( \frac{\gamma _{\text{WLT}}}{\snr}\left[ \frac{1}{\xi _n}-\frac{\eta _{\text{WL},n}}{\gamma _{\text{WLT}}} \right]^+ \right) \right\}\label{Pmmse} 
		\end{split}
    \end{align}
    \end{subequations}
	where $\gamma _{\text{WLT}}=2^{2R}-1$ is the SINR threshold of WL receivers, $F_{\chi _{2M-N+1}^{2}}$ is the CDF of $\chi _{2M-N+1}^{2}$ and $[x]^+=\max(x,0)$. In the PPC case, the outage probability expressions can be simplified to
	\begin{subequations}
	\begin{align}
	\mathcal{P}_{\text{WL}-\text{ZF}}&=F_{\chi _{2M-N+1}^{2}}\left( \frac{\gamma_{\text{WLT}}}{\snr \cdot\xi _{\text{PPC}}} \right)
	\\
	\mathcal{P}_{\text{WL}-\text{MMSE}}&\simeq \mathbb{E}_{\eta _{\text{WL},n}}\left\{ F_{\chi _{2M-N+1}^{2}}\left( \frac{\gamma _{\text{WLT}}}{\snr\cdot \xi _{\text{PPC}}}\left[ 1-\frac{\xi _{\text{PPC}}}{\gamma _{\text{WLT}}}\eta _{\text{WL},n} \right] ^+ \right) \right\} 
	\end{align}
	\end{subequations}
	where the probability density function (PDF) of $\eta _{\text{WL},n}$ can be obtained according to \eqref{e:addterm}.

	\subsection{Comparative Analysis}
\label{s:compana}

We now use the expressions in \eqref{e:outageprob} to study the diversity and coding gains for WL detection. For this, we note that the polynomial approximation of $F_{\chi _{k}^{2}}\left( x \right)$ around $x=0$ is given by 
	\begin{equation}\label{Fexp}
	F_{\chi _{k}^{2}}\left( x \right) =\frac{1}{\left( k/2 \right) 2^{k/2}\Gamma \left( k/2 \right)}x^{k/2}+o\left( x^{k/2} \right),
	\end{equation}
where $\Gamma( a) \triangleq \int_0^{\infty}{e^{-t}t^{a-1}\mathrm{d}t}$ is the Gamma function. Considering the high SNR regime such that $\frac{\gamma _{\text{WLT}}}{\snr}\rightarrow 0$ and letting $d_{\text{WL}}=2M-\left( N-1 \right) /2$, we can expand the outage probability expressions of WL-ZF and WL-MMSE receivers using (\ref{Fexp}) in \eqref{Pzf} and \eqref{Pmmse}, as follows
\begin{align*}	
	\begin{split}
	\mathcal{P}_{\text{WL}-\text{ZF}}&\simeq \mathbb{E}_{\xi _n}\left\{ \frac{1}{d_{\text{WL}}2^{d_{\text{WL}}}\Gamma \left( d_{\text{WL}} \right)}\left( \frac{\gamma _{\text{WLT}}}{\snr\cdot\xi _n} \right) ^{d_{\text{WL}}} \right\}
	\\
	&=\left( \left[ \mathbb{E}_{\xi _n}\left\{ \frac{1}{d_{\text{WL}}2^{d_{\text{WL}}}\Gamma \left( d_{\text{WL}} \right)}\left( \frac{\gamma _{\text{WLT}}}{\xi _n} \right) ^{d_{\text{WL}}} \right\} \right] ^{-1/d_{\text{WL}}}\cdot\snr \right) ^{-d_{\text{WL}}}
	\\
	\end{split}
	\\
	\begin{split}
	\mathcal{P}_{\text{WL}-\text{MMSE}}&\simeq \mathbb{E}_{\xi _n,\eta _{\text{WL},n}}\left\{ \frac{1}{d_{\text{WL}}2^{d_{\text{WL}}}\Gamma \left( d_{\text{WL}} \right)}\left( \frac{\gamma _{\text{WLT}}}{\snr}\left[ \frac{1}{\xi _n}-\frac{\eta _{\text{WL},n}}{\gamma _{\text{WLT}}} \right] ^+ \right) ^{d_{\text{WL}}} \right\}
	\\
	&=\left( \left[ \mathbb{E}_{\xi _n,\eta _{\text{WL},n}}\left\{ \frac{1}{d_{\text{WL}}2^{d_{\text{WL}}}\Gamma \left( d_{\text{WL}} \right)}\left( \gamma _{\text{WLT}}\left[ \frac{1}{\xi _n}-\frac{\eta _{\text{WL},n}}{\gamma _{\text{WLT}}} \right] ^+ \right) ^{d_{\text{WL}}} \right\} \right] ^{-1/d_{\text{WL}}}\cdot \snr \right) ^{-d_{\text{WL}}}
	\end{split}
\end{align*}
It is clear that the outage diversity and coding gains of WL-ZF and WL-MMSE receivers are expressed as
	\begin{subequations}\label{WL-Gains}
	\begin{align}
	d_{\text{WL}-\text{ZF}}&=d_{\text{WL}-\text{MMSE}}=d_{\text{WL}}
	\\
	C_{\text{WL}-\text{ZF}}&=\frac{2\left( d_{\text{WL}}\Gamma( d_{\text{WL}}) \right)^{1/d_{\text{WL}}}}{\gamma_{\text{WLT}}}T_{\text{WL}-\text{ZF}}\label{GainZf}
	\\
	C_{\text{WL}-\text{MMSE}}&=\frac{2\left( d_{\text{WL}}\Gamma( d_{\text{WL}} ) \right)^{1/d_{\text{WL}}}}{\gamma_{\text{WLT}}}T_{\text{WL}-\text{MMSE}}\;,\label{GainMMSE}
	\end{align}
	\end{subequations}
	where
	\begin{subequations}
	\begin{align}
	T_{\text{WL}-\text{ZF}} &=\left[ \mathbb{E}_{\xi _n}\left\{ \xi _{n}^{-d_{\text{WL}}} \right\} \right] ^{-1/d_{\text{WL}}}\label{Twl-zf}
	\\
	T_{\text{WL}-\text{MMSE}}&=\left[ \mathbb{E}_{\xi _n,\eta _{\text{WL},n}}\left\{ \left( \left[\frac{1}{\xi _n}-\frac{\eta _{\text{WL},n}}{\gamma _{\text{WLT}}} \right]^+\right) ^{d_{\text{WL}}} \right\} \right] ^{-1/d_{\text{WL}}}\label{Twl-mmse}
	\end{align}
	\end{subequations}

	We observe that the WL-MMSE receiver does not provide an extra diversity gain over the WL-ZF receiver but rather yields an increase in coding gain since $T_{\text{WL}-\text{MMSE}}>T_{\text{WL}-\text{ZF}}$. As $\gamma _{\text{WLT}}\rightarrow \infty$, we have $T_{\text{WL}-\text{MMSE}}\rightarrow T_{\text{WL}-\text{ZF}}$, i.e., in a high data rate scenario, the WL-MMSE receiver has no significant outage performance superiority over the WL-ZF receiver. In contrast, in a low rate scenario where $\gamma _{\text{WLT}}\rightarrow 0$, the WL-MMSE receiver is expected to have a notably increased coding gain when compared to the WL-ZF receiver. It should be noted that these insights were presented using heuristic arguments in our previous work \cite{GuiBalasubramanyaConnectivity19}.
	
	For the purpose of a WL-vs.-CL comparison, we present the diversity and coding gains of CL-ZF and CL-MMSE receivers for decoding data streams generated by complex-valued modulations\footnote{Since the diversity and coding gains of CL-ZF/-MMSE receivers with complex-valued modulation can be derived using the same approach as that for deriving (\ref{WL-Gains}a)-(\ref{WL-Gains}c), their intermediate steps are omitted due to the space limit. Note that the outage probability expression for complex-valued transmission should be applied to the output SINR of CL-ZF/MMSE receivers derived from the complex-valued model \eqref{MimoMod}.}
	\begin{subequations}
		\begin{align}
		d_{\text{CL}-\text{ZF}}&=d_{\text{CL}-\text{MMSE}}=d_{\text{CL}}=M-( N-1)
		\\
		C_{\text{CL}-\text{ZF}}&=\frac{\left( d_{\text{CL}}\Gamma( d_{\text{CL}}) \right)^{1/d_{\text{CL}}}}{\gamma_{\text{CLT}}}T_{\text{CL}-\text{ZF}}
		\\
		C_{\text{CL}-\text{MMSE}}&=\frac{\left( d_{\text{CL}}\Gamma( d_{\text{CL}} ) \right)^{1/d_{\text{CL}}}}{\gamma_{\text{CLT}}}T_{\text{CL}-\text{MMSE}}\;,
		\end{align}
	\end{subequations}
	where $\gamma _{\text{CLT}}=2^R-1$ and
	\begin{subequations}
		\begin{align}
		T_{\text{CL}-\text{ZF}}&=\left[ \mathbb{E}_{\xi _n}\left\{ \xi_{n}^{-d_{\text{CL}}} \right\} \right] ^{-1/d_{\text{CL}}}
		\\
		T_{\text{CL}-\text{MMSE}}&=\left[ \mathbb{E}_{\xi _n,\eta _{\text{CL},n}}\left\{ \left( \left[\frac{1}{\xi _n}-\frac{\eta _{\text{CL},n}}{\gamma_{\text{CLT}}} \right]^+\right) ^{d_{\text{CL}}} \right\} \right]^{-1/d_{\text{CL}}}
		\end{align}
	\end{subequations}
	with $\eta_{\text{CL},n}$ being the SINR difference between the CL-ZF and CL-MMSE receivers. In the PPC case, the additional term $\eta_{\text{CL},n}$ is a scaled $F$-distributed random variable with DoF parameters $d_1=2( N-1) $ and $d_2=2( M-N+2)$ \cite{JiangVaranasiPerformance11}.
	
	For a fair comparison, we assume that the number of receive antennas $M$ at the BS and the target rate $R$ for each user are the same for the WL and CL systems. Denoting the number of users by $N_{\text{WL}}$ and $N_{\text{CL}}$ for the two system paradigms, we can make the following observations.
	\begin{itemize}
		\item For $N_{\text{WL}}=N_{\text{CL}}$, we have $d_{\text{WL}}>d_{\text{CL}}$, i.e., for detection of the same number of users, the diversity gain of WL receivers is greater than that of CL receivers. Therefore, in the high SNR regime, the WL system will have a better error-rate performance than the CL system. 
		
        \item For $N_{\text{WL}}< 2N_{\text{CL}}-1$, we have $d_{\text{WL}}> d_{\text{CL}}$. This means that the number of users can be nearly doubled with WL detection, while having an improved diversity gain.
		
		\item For $N_{\text{WL}} = 2N_{\text{CL}}-1$, we have $d_{\text{WL}}= d_{\text{CL}}$ and
		\begin{equation}\label{e:ratioofcodinggains}
		\frac{C_{\text{WL}-\text{ZF}}}{C_{\text{CL}-\text{ZF}}}=\frac{2\gamma _{\text{CLT}}}{\gamma _{\text{WLT}}}=\frac{2\left( 2^R-1 \right)}{2^{2R}-1} \triangleq \mathcal{L}( R ).
		\end{equation}
		Since $\mathcal{L}( R) <1$, the WL-ZF receiver suffers from a decrease in coding gain in this case. 
                However, as $\mathcal{L}( R) \rightarrow 1$ for $R\rightarrow  0$, the decrease in coding gain of the WL-ZF receiver becomes negligible for lower target rates. 
	
		Since the MMSE receivers converge to the ZF receivers in the high data rate regime, it follows that
		\begin{equation}
		\frac{C_{\text{WL}-\text{MMSE}}}{C_{\text{CL}-\text{MMSE}}}\approx \frac{C_{\text{WL}-\text{ZF}}}{C_{\text{CL}-\text{ZF}}}=\mathcal{L}( R), \ \ \text{for~large} \ R .
		\end{equation}
		Hence, the WL-MMSE receiver also has a decreased coding gain compared to the CL-MMSE counterpart. When $R$ is small, i.e., in the low rate regime, it is difficult to analytically compare $C_{\text{WL}-\text{MMSE}}$ and $C_{\text{CL}-\text{MMSE}}$. Through the numerical results, it will be shown that the decrease in coding gain for the WL-MMSE receiver also becomes negligible in this case.
		
		\item As stipulated from the WL MIMO channel model earlier, the largest number of users supported by the WL receivers with a positive diversity gain is $N_{\text{WL}}=2M$. 
		
	\end{itemize}
	
	\section{Outage Performance of WL-SIC Receivers}
\label{s:performancesic}

	In this section, we enhance the WL-ZF and WL-MMSE receivers through the application of SIC and analyze their outage performances in the high SNR regime. In order to facilitate the outage probability analysis for WL-SIC receivers, we first derive the polynomial approximation of the marginal distribution of ordered eigenvalues for real Wishart matrices .
	
	\subsection{Ordered-Eigenvalue Distribution of Real Wishart Matrix}
	Let $\ve{X}\in \mathbb{R}^{n\times m}$ be a matrix with its entries drawn  i.i.d.\ from a real-valued Gaussian distribution with zero mean and unit variance and assume $n \le m $. Furthermore, let $\lambda _k$ for $1\le k \le n$ denote the sorted eigenvalues of the real central Wishart matrix $\ve{XX}^T$ with  $0\le \lambda _1\le \lambda _2\le \cdots \le \lambda _n<\infty$.
	
\begin{theorem}\label{PolExpofCdf-k}
	The polynomial expansion of the marginal CDF of $\lambda _k$ around $\lambda_k=0$ can be expressed as
	\begin{equation}
\label{e:polyexp}
	\mathrm{Pr}(\lambda_k<\epsilon) =\beta_k\epsilon ^{d_k}+o(\epsilon ^{d_k}),
	\end{equation}
	where $\beta_k$ is a coefficient independent of $\epsilon$ and 
	\begin{equation}
	d_k =\frac{1}{2}k\left( m-n+k \right).
	\end{equation}
	Hence, the marginal CDF of $\lambda _k$ satisfies the following asymptotical property:
	\begin{equation}
	\lim_{\epsilon \rightarrow 0^+}\frac{\log\mathrm{Pr}(\lambda_k<\epsilon)}{\log \epsilon}=d_k ,\,\,1\le k\le n.
	\end{equation}
\end{theorem}

The proof is provided in Appendix~\ref{p:theorem1}. Theorem~\ref{PolExpofCdf-k} has complemented the polynomial approximation of the marginal CDF of the ordered eigenvalues for complex-valued Wishart matrices provided in \cite{OrdonezPalomarHigh07}.

\begin{theorem}\label{PolExpofCdf-1}
	The polynomial expansion \eqref{e:polyexp} for the marginal CDF of the smallest eigenvalue $\lambda_1$ has the parameters
	\begin{subequations}
		\begin{align}
		d_1&=\frac{1}{2}( m-n+1)\label{e:d1def} \\
		\beta_1&=K_{nm}^{-1}d_{1}^{-1}\sqrt{|\ve{J}|}\label{e:betadef}
		\end{align}
	\end{subequations}
	with $K_{nm}=\left( \frac{2^m}{\pi} \right) ^{\frac{n}{2}}\prod_{i=1}^n{\Gamma \left( \frac{m-i+1}{2} \right) \Gamma \left( \frac{n-i+1}{2} \right)}$. For $n$ odd, the skew-symmetric matrix $\ve{J}$ is $(n-1)\times (n-1)$ with entries
	\begin{equation}\label{Jij}
[\ve{J}]_{i,j}=-[\ve{J}]_{j,i}=2^{b_i+b_j+1}\Gamma(b_i) \Gamma(b_j) \sum_{k=1}^{j-i}{\frac{2^{-( b_i+b_j-k)}\Gamma(b_i+b_j-k)}{\Gamma(b_i) \Gamma(b_{j-k+1})}}
	\end{equation}
	where $1\le i < j\le n-1$ and $b_i =\frac{1}{2}( m-n+1) +i$. For $n$ even, the skew-symmetric matrix $\ve{J}$ is $n\times n$ with $[\ve{J}]_{i,j}$ given by (\ref{Jij}) for $1\le i < j\le n-1$ and the additional entries given by $[\ve{J}]_{i,n}=-[\ve{J}]_{n,i}=2^{b_i}\Gamma(b_i)$ for $1\le i \le n-1$ and $[\ve{J}]_{n,n}=0$.
\end{theorem}	

	The proof is provided in Appendix~\ref{p:theorem2}.

	\subsection{Asymptotic Outage Probability}
	Following the analysis of diversity-multiplexing tradeoff for CL-SIC receivers in \cite{JiangVaranasiPerformance11}, one can show that, for any ordering rule, the outage probability of WL-SIC receivers for the $n^\mathrm{th}$ user is asymptotically equal to that for the first-layer detection. Hence, for the $n^\mathrm{th}$ user, the outage probability of WL-ZF-SIC and WL-MMSE-SIC receivers with SINR-maximization ordering is given by
	\begin{subequations}\label{Psic}
		\begin{align}
		\mathcal{P}_{\text{WL}-\text{ZSIC},n}&\simeq \Pr\left( \max_{1\le i\le N}\left\{ \gamma_{\text{WL}-\text{ZF},i} \right\} \le \gamma_{\text{WLT}} \right) \label{e:Psica}
		\\
		\mathcal{P}_{\text{WL}-\text{MSIC},n}&\simeq \Pr\left( \max_{1\le i\le N}\left\{ \gamma_{\text{WL}-\text{MMSE},i} \right\} \le \gamma_{\text{WLT}} \right).\label{e:Psicb}
		\end{align}
	\end{subequations}
	If the output SINRs were independent, we could readily calculate $\mathcal{P}_{\text{WL}-\text{ZSIC,}n}$ and $\mathcal{P}_{\text{WL}-\text{MSIC,}n}$ using the simple order statistics for independent random variables as in \cite{DingYangOn14}. However, from (\ref{ZFSINR}b) and (\ref{MMSESINR}b), we observe that the output SINRs are highly correlated, and thus a different approach is needed. 
	
	\subsubsection{WL-ZF-SIC Receiver} For convenience, we introduce 
	\begin{equation}
	\omega_n\triangleq \left[ \left( 2\ve{H}^T\ve{H} \right) ^{-1} \right] _{n,n},
	\end{equation}
such that (\ref{ZFSINR}a) becomes
	\begin{equation}
	\gamma_{\text{WL}-\text{ZF,}n}=\snr\frac{\xi _n}{\omega _n}.
	\end{equation}
		Let $2\ve{H}^T\ve{H}=\ve{V}\ve{\varLambda V}^T$ be the eigenvalue decomposition, where $\ve{V}\in \mathbb{R}^{N\times N}$ is an orthonormal matrix 
 and $\ve{\varLambda }=\text{diag}\left\{ \lambda_1,\lambda_2,\ldots,\lambda_N \right\}$ is diagonal with the ordered eigenvalues $\lambda_1 \le \cdots \le \lambda_N$. Since
	\begin{equation}\label{omega}
	\omega_n=\left[ \ve{V\varLambda }^{-1}\ve{V}^T \right] _{n,n}=\ve{v}_{n}^{T}\ve{\varLambda }^{-1}\ve{v}_n=\sum_{i=1}^N{v_{n,i}^{2}\lambda_{i}^{-1}}\ge v_{n,1}^{2}\lambda _{1}^{-1},
	\end{equation}
	where $\ve{v}_n$ is the $n^{\mathrm{th}}$ column of $\ve{V}^T$ and $v_{n,i}$ is the $i^{\mathrm{th}}$ entry of $\ve{v}_n$. Then, the output SINR of the WL-ZF receiver is upper bounded as
	\begin{equation}
	\snr\frac{\xi _n}{\omega_n}\le \snr\cdot\lambda _1\frac{\xi _n}{v_{n,1}^{2}}.
	\end{equation}
	Hence, we have the lower bound 
	\begin{equation}
\mathcal{P}_{\text{WL}-\text{ZF},n}\ge \mathrm{Pr}\left( \snr\cdot \lambda _1\frac{\xi _n}{v_{n,1}^{2}}\le \gamma _{\text{WLT}} \right).
	\end{equation}
	Using Theorem~\ref{PolExpofCdf-k} 
 and a proof similar to that for the CL case \cite[Lemma~VI.1]{JiangVaranasiPerformance11} , we can show that the lower bound above is asymptotically tight, i.e.,
	\begin{equation}\label{Pzfasym}
	\mathcal{P}_{\text{WL}-\text{ZF},n}\simeq \mathrm{Pr}\left( \snr\cdot \lambda _1\frac{\xi _n}{v_{n,1}^{2}}\le \gamma _{\text{WLT}} \right).
	\end{equation}
	Hence, the asymptotic outage performance of the WL-ZF receiver is determined by the smallest eigenvalue of real Wishart matrix $2\ve{H}^T\ve{H}$. Applying this result to \eqref{e:Psica} and defining $\theta_n\triangleq v_{n,1}^{2}/\xi_n$ and $\theta _{\min}\triangleq \min_n\{\theta_n\}$, we obtain the asymptotic outage probability for the WL-ZF-SIC receiver as 
	\begin{equation}\label{Pzfsic}
	\mathcal{P}_{\text{WL}-\text{ZSIC},n}\simeq \mathrm{Pr}\left( \max_n\left\{ \snr\frac{\lambda_1}{\theta_n} \right\} \le \gamma_{\text{WLT}} \right)=\mathrm{Pr}\left( \lambda_1\le \frac{\gamma_{\text{WLT}}}{\snr}\theta_{\min} \right) 
	\end{equation}
	
Since $2\ve{H}^T\ve{H}$ is a real central Wishart matrix, the orthonormal matrix $\ve{V}$ is Haar-distributed and independent of $\ve{\varLambda }$ \cite[Sec.~3.2.5]{muirhead:2005}, which implies the statistical independence between $\{ \theta_1,\ldots,\theta_N \}$ and $\{ \lambda_1,\ldots,\lambda_N\}$. From  (\ref{Pzfsic}) and Theorem~\ref{PolExpofCdf-1}, it follows that 
	\begin{equation}
	\mathcal{P}_{\text{WL}-\text{ZSIC},n}\simeq \mathbb{E}_{\theta _{\min}}\left\{ \beta_{\text{WL}}\left( \frac{\gamma_{\text{WLT}}}{\snr}\theta_{\min} \right)^{d_{\text{WL}}} \right\},
	\end{equation}
	where $\beta _{\text{WL}}=K_{N,2M}^{-1}d_{\text{WL}}^{-1}\sqrt{\left| \ve{J}_{\text{WL}} \right|}$ and $\ve{J}_{\text{WL}}$ equals to the matrix $\ve{J}$ in Theorem~\ref{PolExpofCdf-1} with $n=N$ and $m=2M$. From this, the diversity and coding gains follow as 
	\begin{subequations}
		\begin{align}
		d_{\text{WL}-\text{ZSIC}}&=d_{\text{WL}}
		\\
		C_{\text{WL}-\text{ZSIC}}&=\frac{\beta _{\text{WL}}^{-1/d_{\text{WL}}}}{\gamma _{\text{WLT}}}\left[ \mathbb{E}_{\theta _{\min}}\left\{ (\theta _{\min})^{d_{\text{WL}}} \right\} \right] ^{-1/d_{\text{WL}}}.\label{ZF-SICgain}
		\end{align}
	\end{subequations}
    
From (\ref{Pzfasym}) we also obtain an alternative expression of the coding gain of WL-ZF receiver 
	\begin{equation}\label{ZFgain2}
	C_{\text{WL}-\text{ZF}}=\frac{\beta_{\text{WL}}^{-1/d_{\text{WL}}}}{\gamma_{\text{WLT}}}\left[ \mathbb{E}_{\theta _n}\left\{ (\theta _{n})^{d_{\text{WL}}} \right\} \right]^{-1/d_{\text{WL}}},
	\end{equation}
	which can be shown to be equal to (\ref{GainZf}).
	
	Comparing \eqref{ZF-SICgain} and \eqref{ZFgain2}, we observe that the SIC operation with channel-dependent ordering, although not improving the diversity gain,  increases the coding gain. Moreover, the coding gain $C_{\text{WL}-\text{ZSIC}}$ increases with $N$, because the SIC receiver has more candidate users to choose from for ordering.	

	\subsubsection{WL-MMSE-SIC Receiver}It is not tractable to analyze the outage probability of the WL-MMSE-SIC receiver directly from (\ref{MMSESINR}a), since the diagonal entries of $\ve{\varPsi }$ are unequal. Therefore, we turn to the lower and upper bounds of (\ref{MMSESINR}a). In particular, one can easily show that the output SINR (\ref{MMSESINR}a) can be bounded by
	\begin{equation}\label{SinrConstraint}
	\gamma_n^{\text{LB}}\le \gamma_{\text{WL}-\text{MMSE,}n}\le \gamma_n^{\text{UB}},
	\end{equation}
	where
	\begin{subequations}
		\begin{align}
		\gamma_n^{\text{LB}}&=\frac{\snr\cdot \xi_n}{\left[ \left( 2\ve{H}^T\ve{H}+\frac{\xi_{\max}^{-1}}{\snr}\ve{I}_N \right) ^{-1} \right]_{n,n}}-1
		\\
		\gamma_n^{\text{UB}}&=\frac{\snr\cdot \xi_n}{\left[ \left( 2\ve{H}^T\ve{H}+\frac{\xi_{\min}^{-1}}{\snr}\ve{I}_N \right) ^{-1} \right]_{n,n}}-1
		\end{align}
	\end{subequations}
	with $\xi_{\max}\triangleq \max_n\{ \xi_n\}$ and $\xi_{\min}\triangleq \min_n\{ \xi_n\}$. Following the same derivations for the WL-ZF-SIC receiver above, we arrive at
	\begin{equation}
	\mathcal{P}_{\text{WL}-\text{MSIC},n}\simeq \left( C_{\text{WL}-\text{MSIC},n}\snr \right) ^{-d_{\text{WL}}},
	\end{equation}
	where $C_{\text{WL}-\text{MSIC,}n}$ is the coding gain bounded as $C_n^{\text{LB}}\le C_{\text{WL}-\text{MSIC,}n}\le C_n^{\text{UB}}$ with 
\begin{align}\label{ULB}
	C_{n}^{\text{LB}}&=\frac{\beta _{\text{WL}}^{-\text{1/}d_{\text{WL}}}}{\gamma _{\text{WLT}}+1}\left[ \mathbb{E}_{\theta _{\min},\xi _{\max}}\left\{ \left( \left[ \theta _{\min}-\frac{\xi _{\max}^{-1}}{\gamma _{\text{WLT}}+1} \right] ^+ \right) ^{d_{\text{WL}}} \right\} \right] ^{-\text{1/}d_{\text{WL}}}
	\\
	C_{n}^{\text{UB}}&=\frac{\beta _{\text{WL}}^{-\text{1/}d_{\text{WL}}}}{\gamma _{\text{WLT}}+1}\left[ \mathbb{E}_{\theta _{\min},\xi _{\min}}\left\{ \left( \left[ \theta _{\min}-\frac{\xi _{\min}^{-1}}{\gamma _{\text{WLT}}+1} \right] ^+ \right) ^{d_{\text{WL}}} \right\} \right] ^{-\text{1/}d_{\text{WL}}}
	\end{align}
obtained from $\gamma_n^{\text{LB}}$ and $\gamma_n^{\text{UB}}$, respectively.

	Since it is difficult to gain insights into the effect of SIC on the coding gain for WL-MMSE, we consider an alternative path to bound $C_{\text{WL}-\text{MSIC},n}$. Combining (\ref{MMSE-SINR}) and (\ref{Pzfasym}) yields
	\begin{align}
	\label{Pmmsesic}
		\mathcal{P}_{\text{WL}-\text{MSIC},n}&\simeq \Pr\left( \max_n\left\{ \snr\cdot\lambda_1\frac{\xi_n}{v_{n1}^{2}}+\xi_n\eta_{\text{WL},n} \right\} \le \gamma_{\text{WLT}} \right)\\
	  &\label{e:pmmsesic2}
	\le \Pr\left( \lambda_1\le \frac{\gamma_{\text{WLT}} }{\snr}\min_n\left[ \frac{v_{n1}^{2}}{\xi_n}-\frac{\eta_{\text{WL},n}}{ \gamma_{\text{WLT}}} v_{n1}^2\right] \right),
	\end{align}
	where for \eqref{e:pmmsesic2} we upper-bounded the asymptotic outage probability assuming an ordering rule which does not necessarily maximize the SINR. Defining  
	\begin{equation}
           \label{e:varThetamin}
	\vartheta_n\triangleq\left[ \frac{v_{n1}^{2}}{\xi _n}-\frac{\eta_{\text{WL},n}}{\gamma_{\text{WLT}}}v_{n1}^{2} \right],\quad \vartheta_{\min}\triangleq \min\limits_{n}\vartheta_n,
	\end{equation}
        and using Theorem~\ref{PolExpofCdf-1} and the fact that $\lambda _1$ and $\vartheta_{\min}$ are statistically independent, 
 \begin{equation}
 \begin{split}
	\mathrm{Pr}\left( \lambda _1\le \frac{\gamma _{\text{WLT}}}{\snr}\vartheta _{\min} \right) \simeq \mathbb{E}_{\vartheta _{\min}}\left\{ \beta _{\text{WL}}\left( \left[ \frac{\gamma _{\text{WLT}}}{\snr}\vartheta _{\min} \right] ^+ \right) ^{d_{\text{WL}}} \right\}
 \end{split}
 \end{equation}
 Accordingly, the coding gain of the WL-MMSE-SIC receiver is lower bounded as
	\begin{equation}\label{Cwl_msic}
	C_{\text{WL}-\text{MSIC}}\ge \frac{\beta _{\text{WL}}^{-\text{1/}d_{\text{WL}}}}{\gamma _{\text{WLT}}}\left[ \mathbb{E}_{\vartheta _{\min}}\left\{ \left( \left[ \vartheta _{\min} \right] ^+ \right) ^{d_{\text{WL}}} \right\} \right] ^{-\text{1/}d_{\text{WL}}}
	\end{equation}
	For a comparison, coding gain of the WL-MMSE receiver can be re-expressed as
	\begin{equation}\label{Cwl_mmse}
	C_{\text{WL}-\text{MMSE}}=\frac{\beta _{\text{WL}}^{-\text{1/}d_{\text{WL}}}}{\gamma _{\text{WLT}}}\left[ \mathbb{E}_{\vartheta _n}\left\{ \left( \left[ \vartheta _n \right] ^+ \right) ^{d_{\text{WL}}} \right\} \right] ^{-\text{1/}d_{\text{WL}}}
	\end{equation}
	which follows from \eqref{Pmmsesic} without the $\max$-operator. Comparing \eqref{Cwl_msic} and \eqref{Cwl_mmse}, we observe that, similar to the WL-ZF case, SIC with channel-dependent ordering increases the coding gain of the WL-MMSE receiver. The difference to the WL-ZF case is the contribution of the SINR difference term $\eta_{\text{WL,}n}$ in $\vartheta_n$, which creates stronger fluctuation in the SINR. Hence, we expect that coding gain improvement due to SIC is more pronounced for WL-MMSE than for WL-ZF.

	\subsubsection{WL-SIC Receivers With PPC}In the PPC case, we obtain the coding gain of the WL-ZF-SIC receiver as
	\begin{equation}\label{Czsic}
	C_{\text{WL}-\text{ZSIC}}^{\text{PPC}}=\frac{\beta_{\text{WL}}^{-1/d_{\text{WL}}}\xi_{\text{PPC}}}{\gamma_{\text{WLT}}}\left[ \mathbb{E}_{u_{\min}}\left\{ (u_{\min})^{d_{\text{WL}}} \right\} \right]^{-1/d_{\text{WL}}},
	\end{equation}
	where $u_{\min}\triangleq \min_n\{u_n\}$ and $u_n\triangleq |v_{n,1}|^2$. Since $\xi _{\min}=\xi _{\max}=\xi _{\text{PPC}}$, the lower and upper bounds in \eqref{ULB} coincide, and thus the coding gain of the WL-MMSE-SIC receiver can be calculated as
	\begin{equation}
	C_{\text{WL}-\text{MSIC}}^{\text{PPC}}=\frac{\beta _{\text{WL}}^{-\text{1/}d_{\text{WL}}}\xi _{\text{PPC}}}{\gamma _{\text{WLT}}+1}\left[ \mathbb{E}_{u_{\min}}\left\{ \left( \left[ u_{\min}-\frac{1}{\gamma _{\text{WLT}}+1} \right] ^+ \right) ^{d_{\text{WL}}} \right\} \right] ^{-\text{1/}d_{\text{WL}}}.
	\end{equation}
	For such a system, the SIC operation can only exploit the SINR fluctuation from the small-scale fading, and we expect that resulting improvement in coding gain is less significant than that for the system without power control. 
	
	\subsection{Comparative Analysis}
\label{s:performancesiccomparative}
	As the diversity gain does not change when applying SIC, we only compare the coding gains of WL-SIC and CL-SIC receivers. Furthermore, since for $N_{\text{WL}}<2N_{\text{CL}}-1$, the WL-SIC receivers are expected to have a better outage performance than the CL-SIC receivers in the high SNR regime due to $d_{\text{WL}}>d_{\text{CL}}$ (see Section~\ref{s:compana}), we are interested in the case of $N_{\text{WL}}=2N_{\text{CL}}-1$. Finally, for analytical tractability of the comparison, we focus on the PPC case and comment on the general cases at the end of this section.
	
	First, we combine \eqref{GainZf} and (\ref{Czsic}) for the coding gain of WL-ZF with PPC to obtain the identity
	\begin{equation}\label{ZFgain3}
  \beta_{\text{WL}}= 2^{-d_{\text{WL}}}\left( d_{\text{WL}}\Gamma(d_{\text{WL}})\right)^{-1}\left[ \mathbb{E}_{u_n}\{ (u_{n})^{d_{\text{WL}}}\} \right]^{-1},
	\end{equation}
	with which we can further rewrite (\ref{Czsic}) as
	\begin{equation}
\label{e:cpcwl}
	C_{\text{WL}-\text{ZSIC}}^{\text{PPC}}=\frac{2\left( d_{\text{WL}}\Gamma \left( d_{\text{WL}} \right) \right)^{1/d_{\text{WL}}}\xi _{\text{PPC}}}{\gamma _{\text{WLT}}}\left[ \mathbb{E}_{u_n}\{ (u_{n})^{d_{\text{WL}}}\} \right]^{1/d_{\text{WL}}}\left[ \mathbb{E}_{u_{\min}}\{ (u_{\min})^{d_{\text{WL}}}\} \right]^{-1/d_{\text{WL}}}.
	\end{equation}
	Then, following the derivations for the WL-ZF-SIC case, the coding gain for the CL-ZF-SIC receiver is obtained as
	\begin{equation}
\label{e:cpclin}
	C_{\text{CL}-\text{ZSIC}}^{\text{PPC}}=\frac{\left(d_{\text{CL}}\Gamma(d_{\text{CL}}) \right)^{1/d_{\text{CL}}}\xi_{\text{PPC}}}{\gamma_{\text{CLT}}}\left[ \mathbb{E}_{\mu_n}\{(\mu_{n})^{d_{\text{CL}}}\} \right]^{1/d_{\text{CL}}}\left[\mathbb{E}_{\mu_{\min}}\{(\mu_{\min})^{d_{\text{CL}}} \} \right]^{-1/d_{\text{CL}}},
	\end{equation}
	where $\mu_n=|\nu_{n,1}|^2$, $\mu_{\min}=\min_n\{\mu_n\}$, $\nu_{n,1}$ is the first entry of vector $\ve{\nu}_n$, and $\ve{\nu}_n$ is a unit-length eigenvector of the complex Wishart matrix $\boldsymbol{\bar{H}}^H\boldsymbol{\bar{H}}$.
	
        Relating \eqref{e:cpcwl} and \eqref{e:cpclin} for the case $N_{\text{WL}}=2N_{\text{CL}}-1$ and thus $d_{\text{WL}}=d_{\text{CL}}\triangleq d$, we have 
	\begin{equation}
	\frac{C_{\text{WL}-\text{ZSIC}}}{C_{\text{CL}-\text{ZSIC}}}=\frac{2\gamma_{\text{CLT}}}{\gamma_{\text{WLT}}}
\frac{\left[ \mathbb{E}_{u_n}\{ (u_{n})^{d} \} /\mathbb{E}_{u_{\min}}\{ (u_{\min})^{d} \} \right]^{1/d}}{\left[ \mathbb{E}_{\mu _n}\{ (\mu_{n})^{d} \} /\mathbb{E}_{\mu_{\min}}\{ (\mu_{\min})^{d}\} \right]^{1/d}}.
	\end{equation}
	Furthermore, we show in Appendix~\ref{App-C} that
	\begin{subequations}\label{Moments}
		\begin{align}
\label{e:momentcl}
		\mathbb{E}_{\mu _n}\{ (\mu _{n})^{d}\} /\mathbb{E}_{\mu _{\min}}\{ (\mu _{\min})^{d}\} &=(N_{\text{CL}})^{d},
		\\
 \label{e:momentwl}
	\mathbb{E}_{u_n}\{( u_n) ^d\}/\mathbb{E}_{\mu _{\min}}\{( u_{\min})^d\}&\gtrsim\left( N_{\text{WL}} \right) ^d,
		\end{align}
	\end{subequations}
where $\gtrsim$ means ``asymptotically'' for $N_{\text{WL}}\to\infty$. Thus, we obtain
	\begin{equation}\label{Cratio}
	\frac{C_{\text{WL}-\text{ZSIC}}}{C_{\text{CL}-\text{ZSIC}}}\gtrsim \frac{2\gamma _{\text{CLT}}}{\gamma _{\text{WLT}}}\frac{N_{\text{WL}}}{N_{\text{CL}}}=\frac{2\gamma _{\text{CLT}}}{\gamma _{\text{WLT}}}\frac{2N_{\text{CL}}-1}{N_{\text{CL}}}\approx 2\mathcal{L}\left( R \right).
\end{equation}
As $\mathcal{L}( R ) \rightarrow 1$ for $R \rightarrow 0$ (see \eqref{e:ratioofcodinggains}), we have $C_{\text{WL}-\text{ZSIC}}/C_{\text{CL}-\text{ZSIC}} \gtrsim2$. That is, in low data-rate scenarios, the WL-ZF-SIC receiver has a greater coding gain than the CL-ZF-SIC receiver, which is contrary to our earlier result for the WL-ZF receiver. This can be explained by the fact that the WL receiver has more candidate users to choose from for ordering and thus the coding gain improvement by SIC is considerably higher for WL than for CL receivers. On the other hand, for high data-rate scenarios, it is hard to completely compensate the decrease in coding gain of the WL-ZF receiver, even with the use of SIC. Consequently, $C_{\text{WL}-\text{ZSIC}}<C_{\text{CL}-\text{ZSIC}}$ for large $R$.

In summary, as the WL receivers are capable of detecting more users, they benefit more from SIC than the CL receivers, which is manifested in terms of the improvement in coding gain, especially in low-rate scenarios. This improvement is analogous to the multi-user diversity gain in opportunistic communications \cite{TseViswanathFundamentals05}. For this reason, we  expect that WL-SIC compares even more favourably to CL-SIC for transmission without or with limited power control, since the instantaneous received powers fluctuate more and thus the system provides more diversity.

\section{Numerical Results}
\label{s:results}

In this section, we present numerical performance results to illustrate \textit{(a)} the convergence of the asymptotic results in Theorems~\ref{PolExpofCdf-k} and \ref{PolExpofCdf-1}, \textit{(b)} the advantages for various types of WL receivers compared to their CL counterparts in terms of outage performance, packet drop-out probability and eventually the system throughput.

\subsection{CDFs of Smallest Eigenvalues}
Theorems~\ref{PolExpofCdf-k} and \ref{PolExpofCdf-1} provide the asymptotic result 
\begin{equation}
F_{\lambda_1}(\epsilon)=\Pr(\lambda_1<\epsilon)=\frac{\sqrt{|\ve{J}|}}{K_{nm}\cdot (m-n+1)/2}\epsilon^{(m-n+1)/2}+o(\epsilon^{(m-n+1)/2})
\end{equation}
for the smallest eigenvalue of  real-valued central Wishart matrices with $n\le m$, where $\ve{J}$ and $K_{nm}$ are given in Theorem~\ref{PolExpofCdf-1}. Fig.~\ref{EigCdf} shows the asymptotic approximation together with the empirical result from simulating $\lambda_1$ (curves labeled with $k=1$) for several combinations of $n$ and $m$. For clarity, the $x$-axis is set to be $-\log _{10}(\epsilon)$. We observe that the asymptotic approximations converge well to the empirical results as $\epsilon \rightarrow 0$. Fig.~\ref{EigCdf}  also includes the empirical CDF for the second ($k=2$) smallest eigenvalue in the case of $n=m=2$. We note that this CDF curve runs parallel to the one for the smallest eigenvalue ($k=1$) and $n=3$ and $m=6$, which has the same polynomial order $\frac{1}{2}k\left( m-n+k \right) =2$, as predicted by Theorem~\ref{PolExpofCdf-k}.
\begin{figure}[t]
\centering
		\includegraphics[width=0.65\columnwidth]{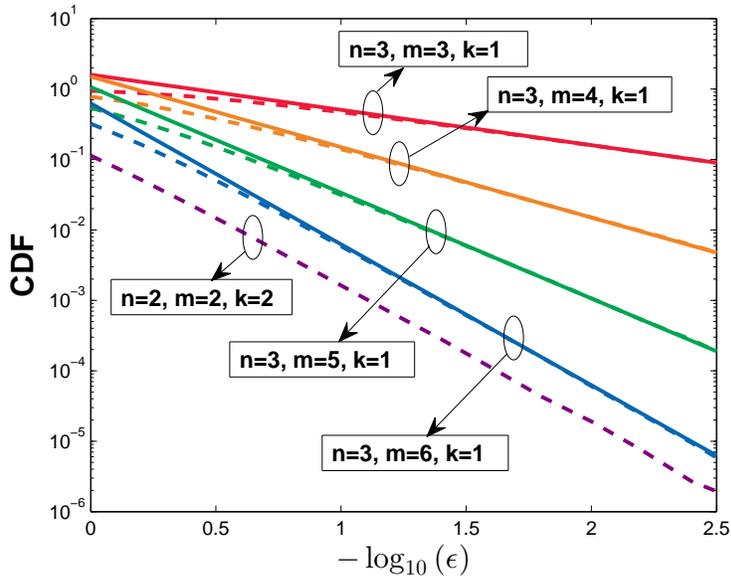}
\vspace*{-5mm}
		\caption{CDF of the $k$th smallest eigenvalue of real-valued central Wishart matrix $\ve{XX}^T$ with $\ve{X}\in \mathbb{R}^{n\times m}$. Solid lines: Asymptotic approximations from Theorems~\ref{PolExpofCdf-k} and \ref{PolExpofCdf-1}. Dashed lines: empirical results from simulation.}\label{EigCdf}
\vspace*{-5mm}
\end{figure}

\subsection{Outage Performance}
\label{outage}
For the outage performance evaluation, we consider  transmission without power control as an illustrative example. We generate the power variation term for the $i^{\mathrm{th}}$ user as
\begin{equation}\label{lsf}
\xi _i=\beta _i+\psi _i\ \text{dB},
\end{equation}
where $\beta_i=-120.9-\text{37.6}\log_{10}(r_i)\ 0\le r_i\le r_{\text{cell}}$ is the pathloss model used in the third Generation Partnership Project (3GPP) standardization \cite{3GPPCellular15} for cellular MTC systems, with $r_i$ being the distance in kilometers between user $i$ and the base station, and $r_{\text{cell}}=0.91$~km the cell radius. The shadowing fading factor $\psi_i$ in dB is a normally distributed random variable with zero mean and standard deviation $8$~dB. We assume uniformly randomly distributed user locations, so that the statistical variation of $\xi_i$ comes from both the shadowing factor $\psi_i$ and the distance-dependent pathloss factor $\beta_i$. The per-user rate $R$ is fixed for both WL and CL receivers, which corresponds to the target data rate of the mMTC user application. Recall that the WL receiver processes one-dimensional modulated signals while the CL receiver works with complex-valued modulation. Thus, to ensure the required rate, if a WL receiver, for instance, adopts a channel code rate of $r$ with BPSK modulation over a TTI equal to $T$, the corresponding CL receiver chosen for comparison will adopt quaternary PSK (QPSK) with a code rate $r/2$ over the same TTI ($T$).

Fig.~\ref{WlrOutProb}(a)  and Fig.~\ref{WlrOutProb}(b) show the analytical and simulated results\footnote{When calculating the asymptotic outage probability of WL receivers, we encounter the statistical expectations in the coding gains \eqref{Twl-zf}, \eqref{Twl-mmse}, \eqref{ZF-SICgain} and \eqref{Cwl_msic} for the WL-ZF, WL-MMSE, WL-ZF-SIC and WL-MMSE-SIC receivers, respectively. These expectations can be quickly computed via the Monte Carlo integration, which needs to generate i.i.d. samples of $\xi _n$, $\eta _{\text{WL},n}$ and $v_{n,1}$ for $1\le n \le N$. The i.i.d. samples of $\xi _n$ and $\eta _{\text{WL},n}$ can be easily generated using \eqref{lsf} and \eqref{SINRdiff2}, respectively. The i.i.d. samples of $\left\{ v_{n,1} \right\}$ can be conveniently drawn from the standard normal distribution (see Appendix~C).} for the outage probability \eqref{Pout} versus transmit SNR \eqref{e:snr} for the investigated WL receivers without power control and with perfect power control (PPC), respectively. The scenario of $N=4$ users with a data rate $R=2$~bits/sec/Hz and a base station with $M=2$ receive antennas is considered. The analytical results for WL-ZF, WL-MMSE and WL-ZF-SIC receivers converge to the simulated curves in performance range of interest, i.e., at the outage probability of about $10^{-1}$ to $10^{-2}$.  Note that for the WL-MMSE-SIC receiver, the asymptotic result is obtained from the lower bound of the coding gain (\ref{Cwl_msic}), and hence the simulated curve is below the asymptotic curve in the high SNR regime. We further observe that, as it has been shown by the analysis in Section~\ref{s:performancesic}, the SIC operation does not provide an extra diversity gain but increases the coding gain. Moreover, as pointed out in  Section~\ref{s:performancesiccomparative}, the improvement in coding gain is more significant for the WL-MMSE than for the WL-ZF detection. 
\begin{figure*}[t]
	\centering
	\subfigure[Without power control]{\includegraphics[width=0.45\columnwidth]{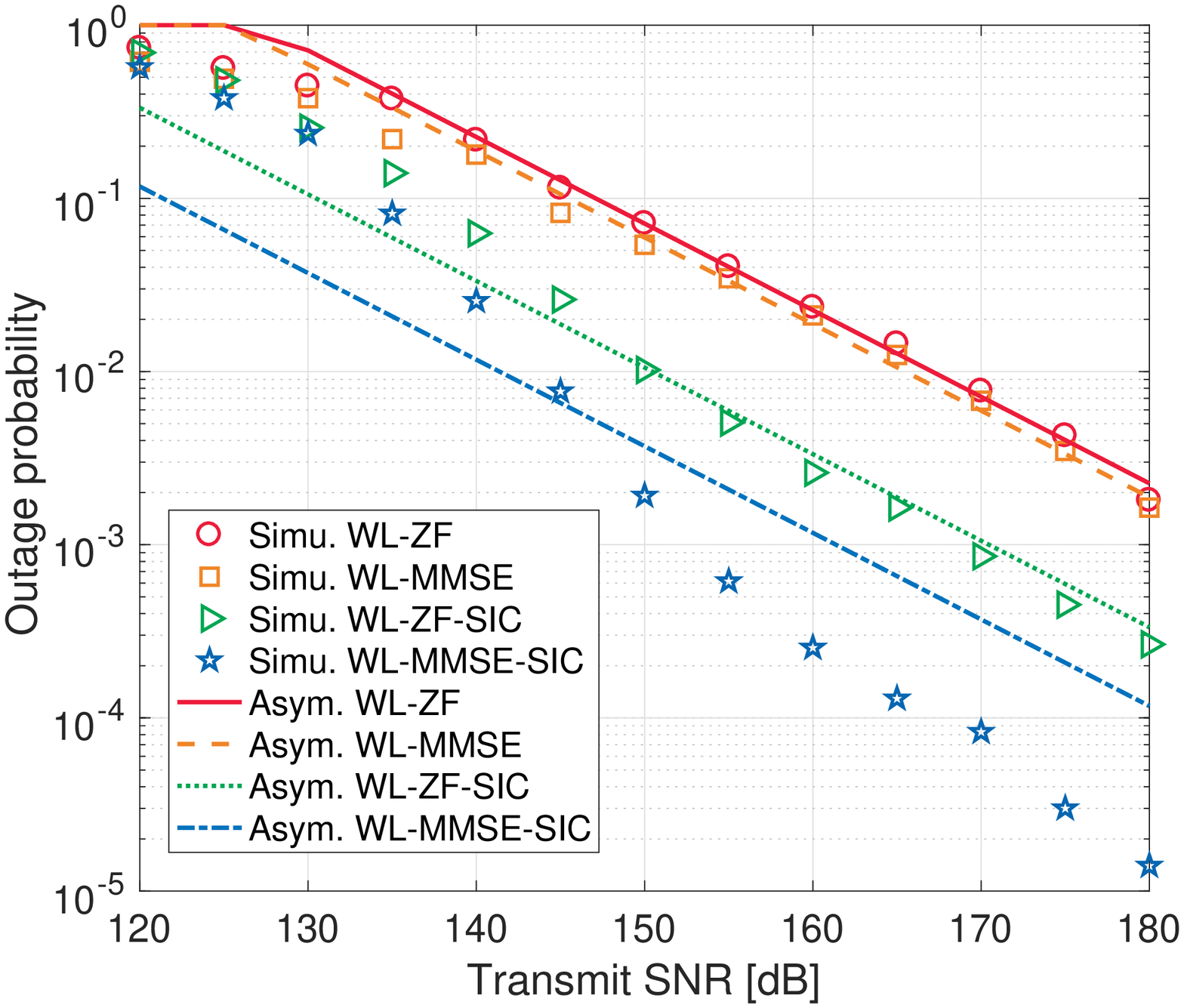}}
	\hspace{0in}
	\subfigure[With perfect power control]{\includegraphics[width=0.48\columnwidth]{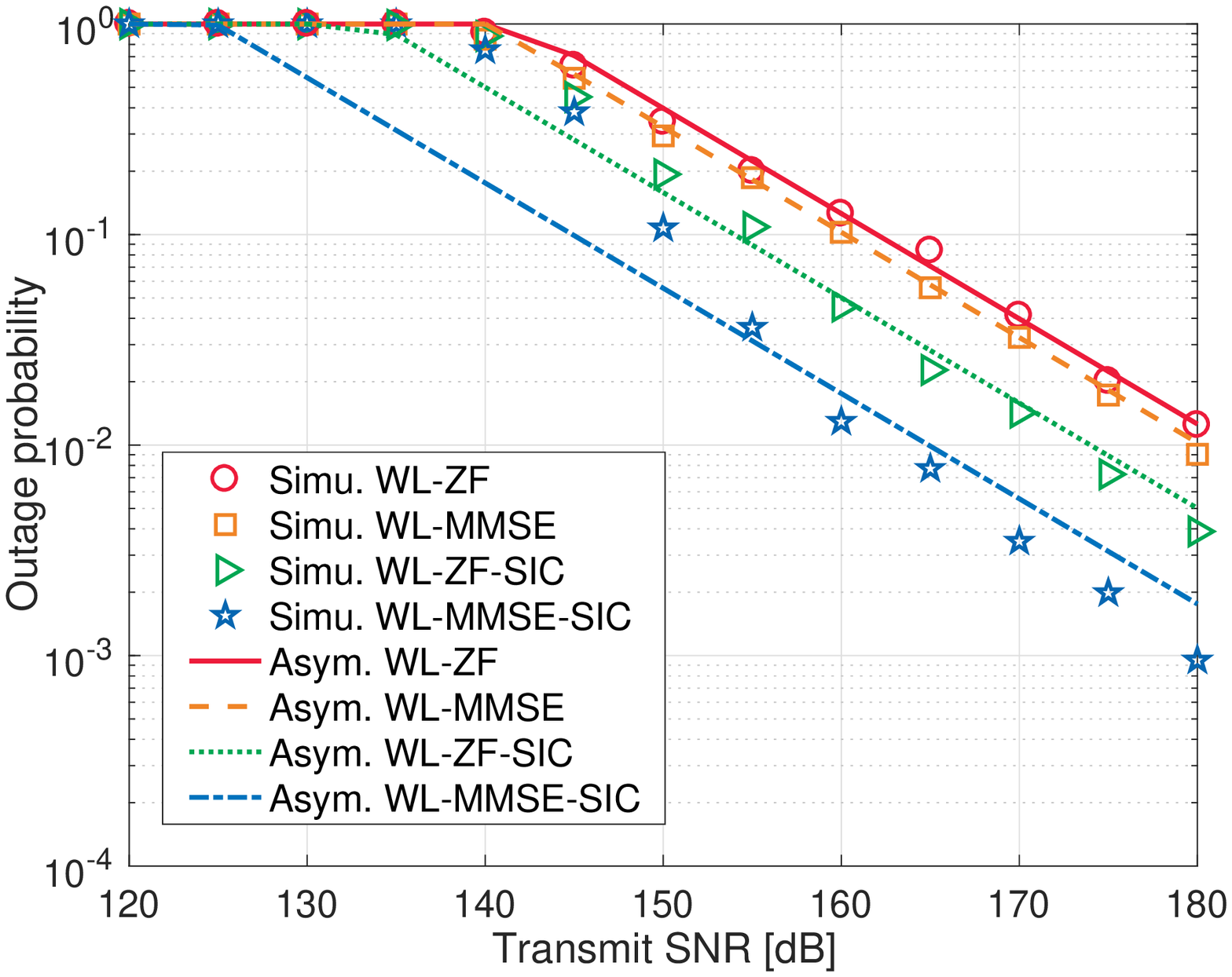}}
	\caption{Outage probability vs.\ transmit SNR for different WL receivers. $N=4$ users with data rate $R=2$~bits/sec/Hz and a base station with $M=2$ receive antennas.}\label{WlrOutProb}
	\vspace*{-5mm}
\end{figure*}


\begin{figure*}[t]
		\centering
		\subfigure[$N_{\text{WL}}=2, N_{\text{CL}}=2, R=2$~bits/sec/Hz]{\includegraphics[width=0.45\textwidth]{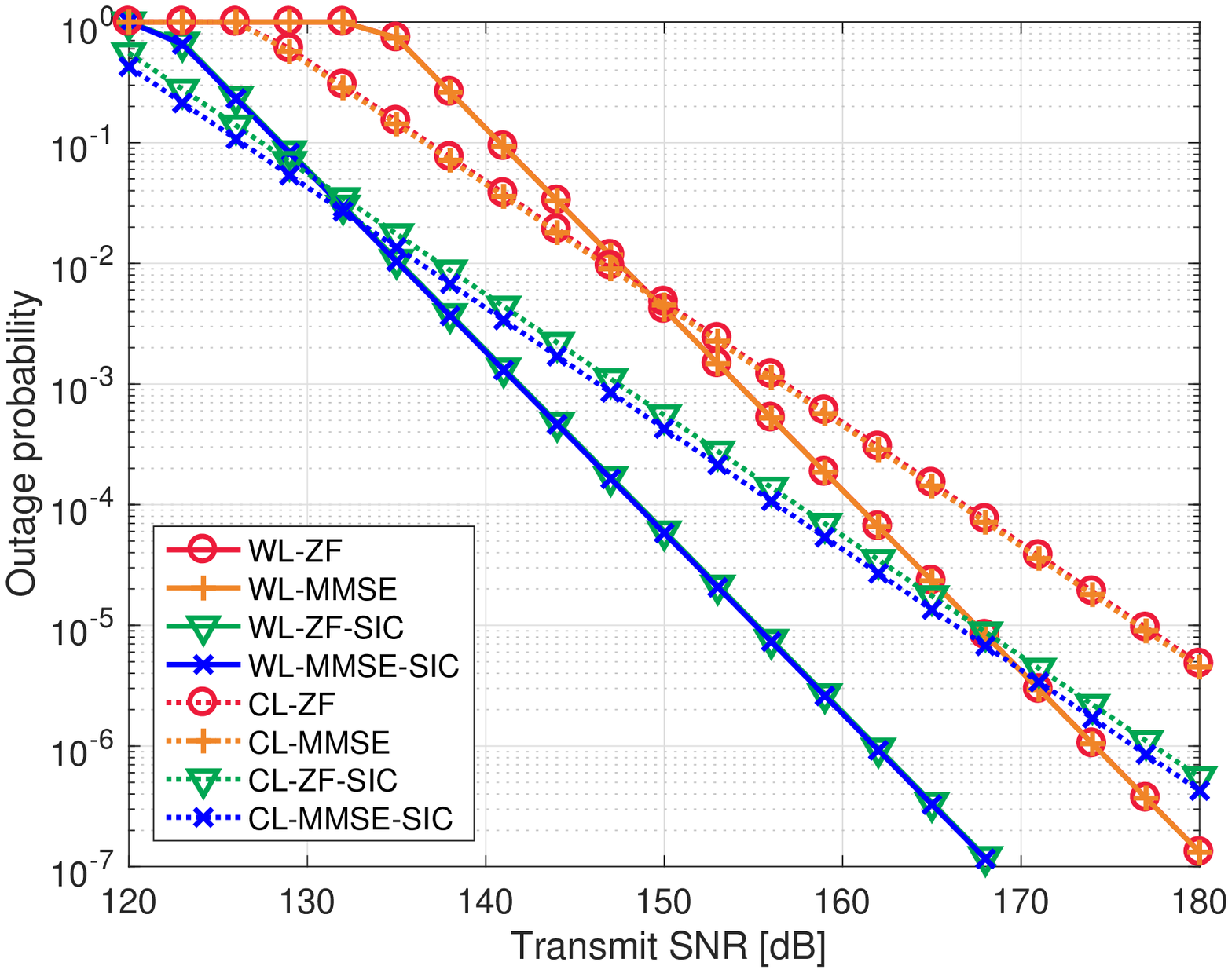}}
		\hspace{0in}
		\subfigure[$N_{\text{WL}}=3, N_{\text{CL}}=2, R=4$~bits/sec/Hz]{\includegraphics[width=0.45\textwidth]{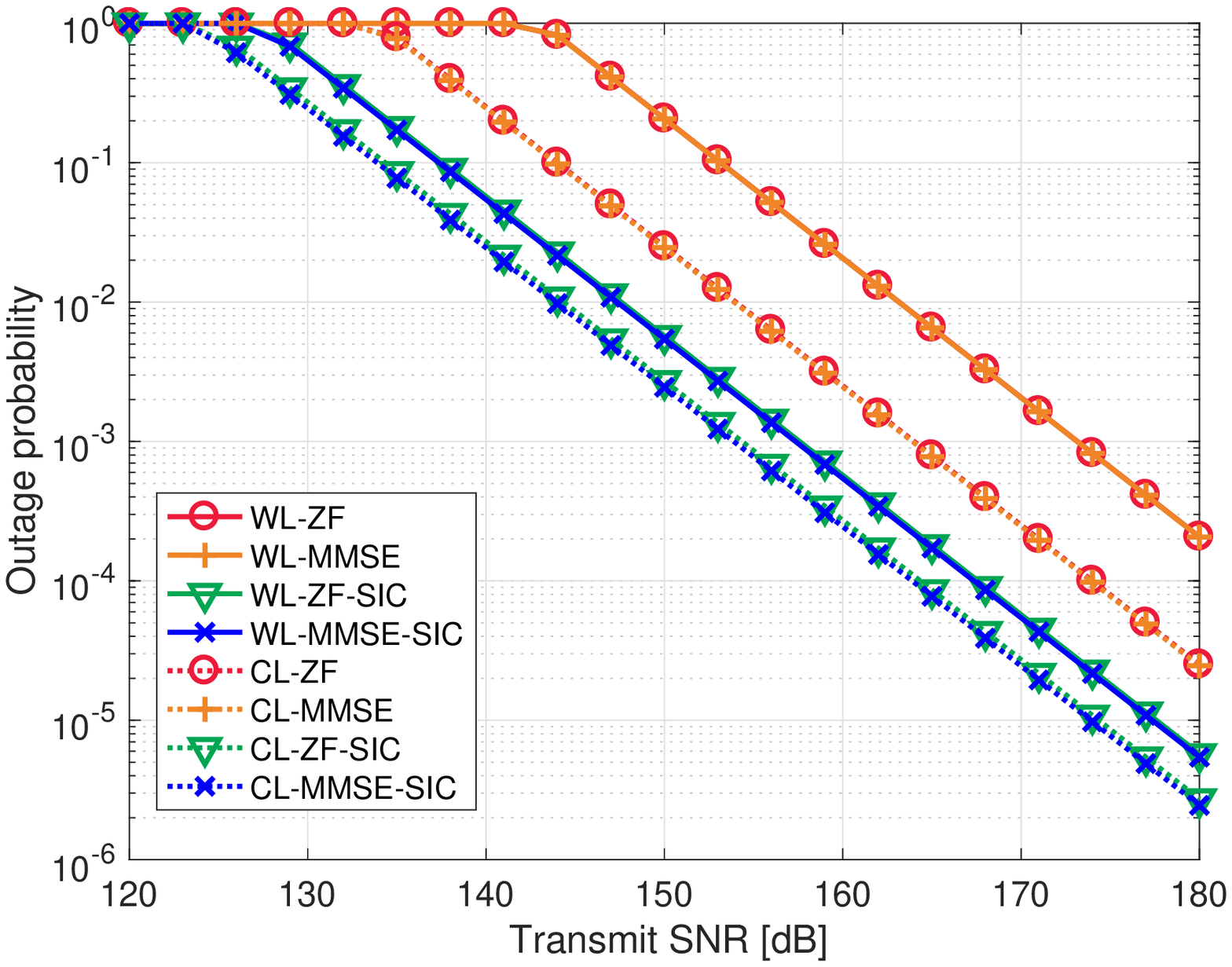}}
		\hspace{0in}
		\subfigure[$N_{\text{WL}}=3, N_{\text{CL}}=2, R=0.3$~bits/sec/Hz]{\includegraphics[width=0.45\textwidth]{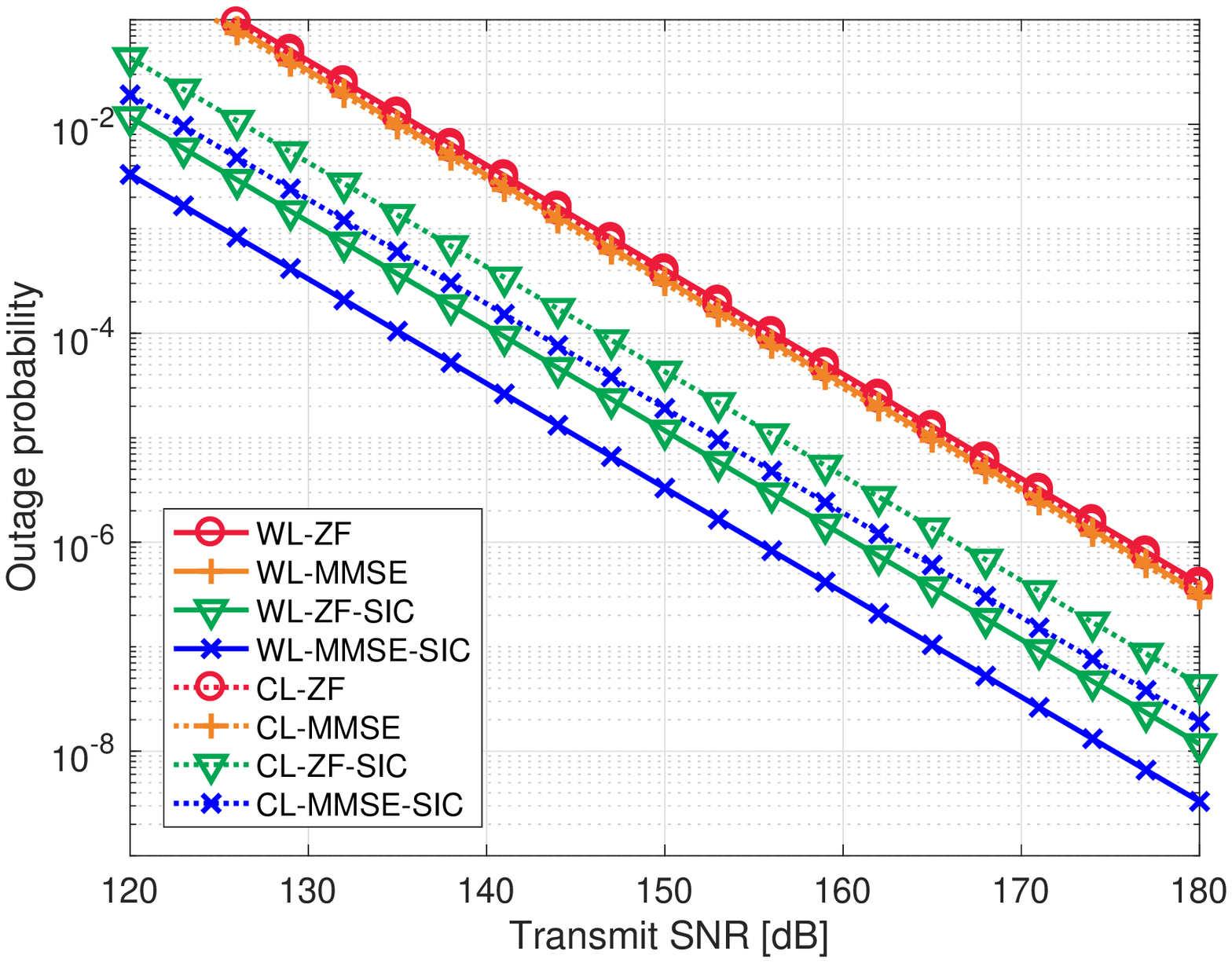}}
		\hspace{0in}
		\subfigure[$N_{\text{WL}}=4, N_{\text{CL}}=2, R=0.3$~bits/sec/Hz]{\includegraphics[width=0.45\textwidth]{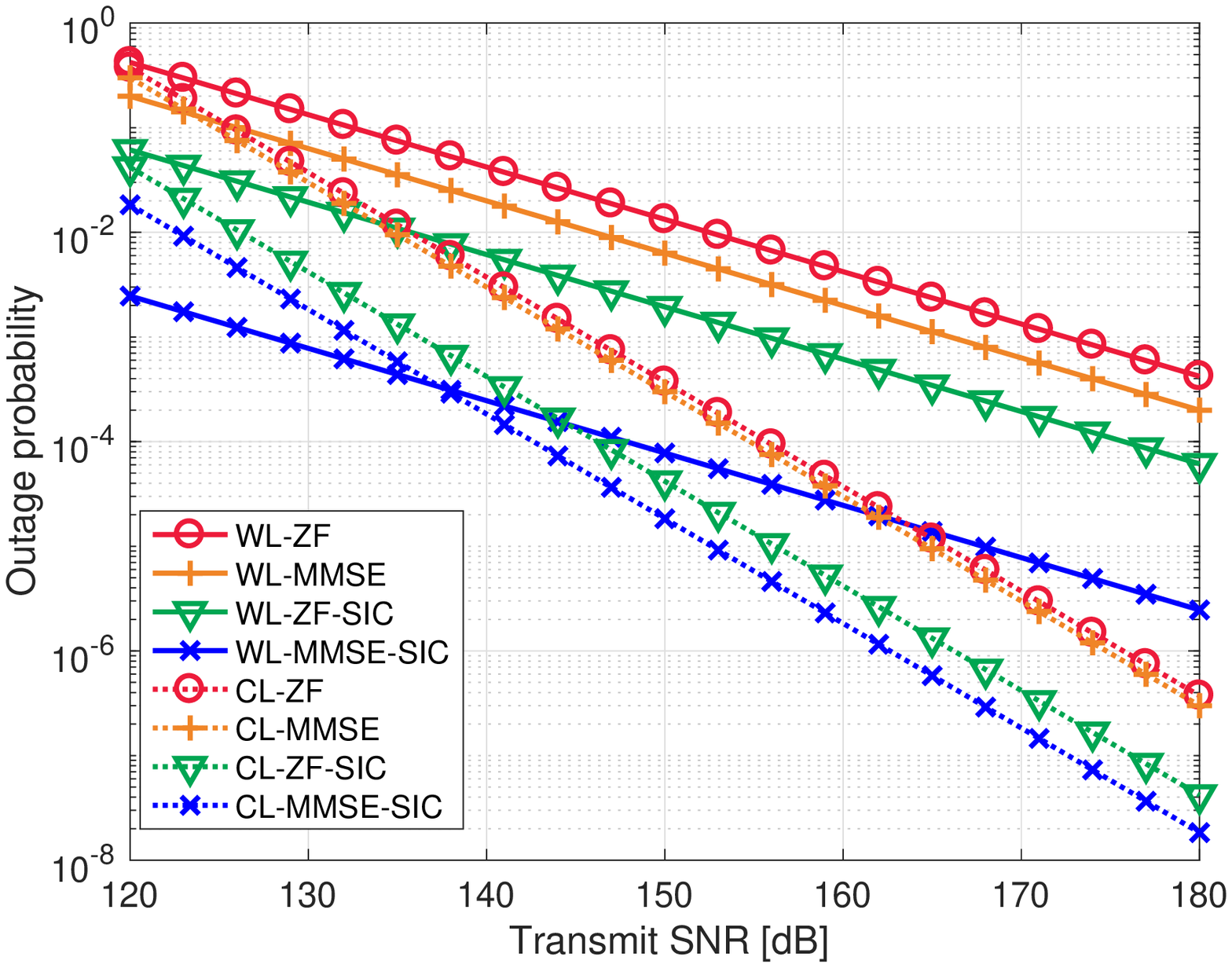}}
		\caption{Asymptotic outage probabilities vs.\ transmit SNR for different  WL and CL receivers. The number of receiver antennas at the base station is fixed to $M=2$. The number of users and the target rates vary. The asymptotic results for both WL-MMSE-SIC and CL-MMSE-SIC receivers are obtained from the corresponding asymptotic lower bounds for the coding gains.}\label{OutProbCmp}
\vspace*{-5mm}
\end{figure*}

Finally, we compare the performance of WL and CL receivers based on asymptotic outage probabilities. Fig.~\ref{OutProbCmp} shows the results for $M=2$ receive antennas and different numbers of users $N_{\text{WL}}$ for the WL and $N_{\text{CL}}$ for the CL case, and different target rates $R$. The asymptotic results for the WL-MMSE-SIC and CL-MMSE-SIC receivers are obtained from the corresponding lower bounds of the coding gain (see (\ref{Cwl_msic}) for the WL case).

For the case of equal number of users, $N_{\text{WL}}=N_{\text{CL}}=2$ in Fig.~\ref{OutProbCmp}(a),  we observe that the diversity gains of all WL receivers are greater than those of the CL receivers. Therefore, one can expect that the WL receivers will provide improved transmission reliability (i.e., lower outage probability) over their CL counterparts at high SNR. 

Next, we consider $N_{\text{WL}}=2N_{\text{CL}} - 1$ such that the diversity gains of WL and CL receivers become identical and compare the asymptotic outage probabilities for different target data rates. Figs.~\ref{OutProbCmp}(b) and~\ref{OutProbCmp}(c) show the results for $N_{\text{WL}}=2N_{\text{CL}} - 1=3$ and a relatively high and a low target rate, respectively. For the high-rate case with $R=4$~bits/sec/Hz, we observe that the WL receivers experience a decrease in coding gain, which is less pronounced for the WL-SIC receivers. The situation changes  when the target data rate is decreased to $R=0.3$~bits/sec/Hz. For WL-ZF and WL-MMSE receivers, the decrease in coding gain compared to CL-ZF and CL-MMSE becomes negligible (see Section~\ref{s:compana}). 
 Consequently, the WL-SIC receivers obtain even improved coding gains compared to those for the CL-SIC counterparts, which was predicted by the discussion in Section~\ref{s:performancesiccomparative}. 

In Fig.~\ref{OutProbCmp}(d), the number of WL users is further increased to $N_{\text{WL}} = 2N_{\text{CL}}=4$, while the target data rate is kept to be $R=0.3$~bits/sec/Hz. We observe that in this case the the WL receivers have a lower diversity gain than the CL receivers, which is clearly undesirable at high SNR. We remark that for all the comparisons above, the rates per user were chosen identical for WL and CL cases and thus the overall system rate is constant when $N_{\text{WL}} = N_{\text{CL}}$ and higher for the WL scenarios when $N_{\text{WL}} > N_{\text{CL}}$. In summary, the numerical results at high SNR match well with our asymptotic analysis for both high and low rate scenarios.

The outage probability analysis can be used to determine the packet drop-out probability and hence the system throughput \cite{GuiBalasubramanyaConnectivity19}, which is described in the following subsection.

\subsection{Packet Drop-out Probability and System Throughput}

\subsubsection{Simulation set-up}
The simulation set-up for characterizing the performance of the proposed WL receivers corresponds to an NB-IoT scenario operating over a system bandwidth of 180~kHz (48 tones with a subcarrier spacing of 3.75~kHz). We adopt a grant-free access mechanism, where users access the network in a contention based manner. We consider open loop power control where each user transmits at the maximum allowable transmit power (23~dBm) over a single subcarrier. The path-loss/shadowing follows the 3GPP path-loss model in [Annex D, \cite{3GPPCellular15}]. The TTI for a WL user is taken to be 32~ms for a transport block (packet) size is 32 bits. The data arrival process is Poisson with a rate $\lambda_a = 4.16 \times 10^{-6}$ packets/TTI/user. The target data rate for each user is $R = 0.3$~bits/s/Hz. The quality of a link is considered to be good if the packet drop probability is $\leq$1\%. As mentioned in Section~\ref{outage}, the WL and CL users adopt BPSK and QPSK, respectively, and they transmit for the same TTI. Hence, the code rate of the WL users is set to $r$, while that of the CL users is set to $r/2$ (in order to maintain the same target rate $R$).

\begin{figure*}[t]
\subfigure[Packet drop probability]{\includegraphics[width=0.5\textwidth]{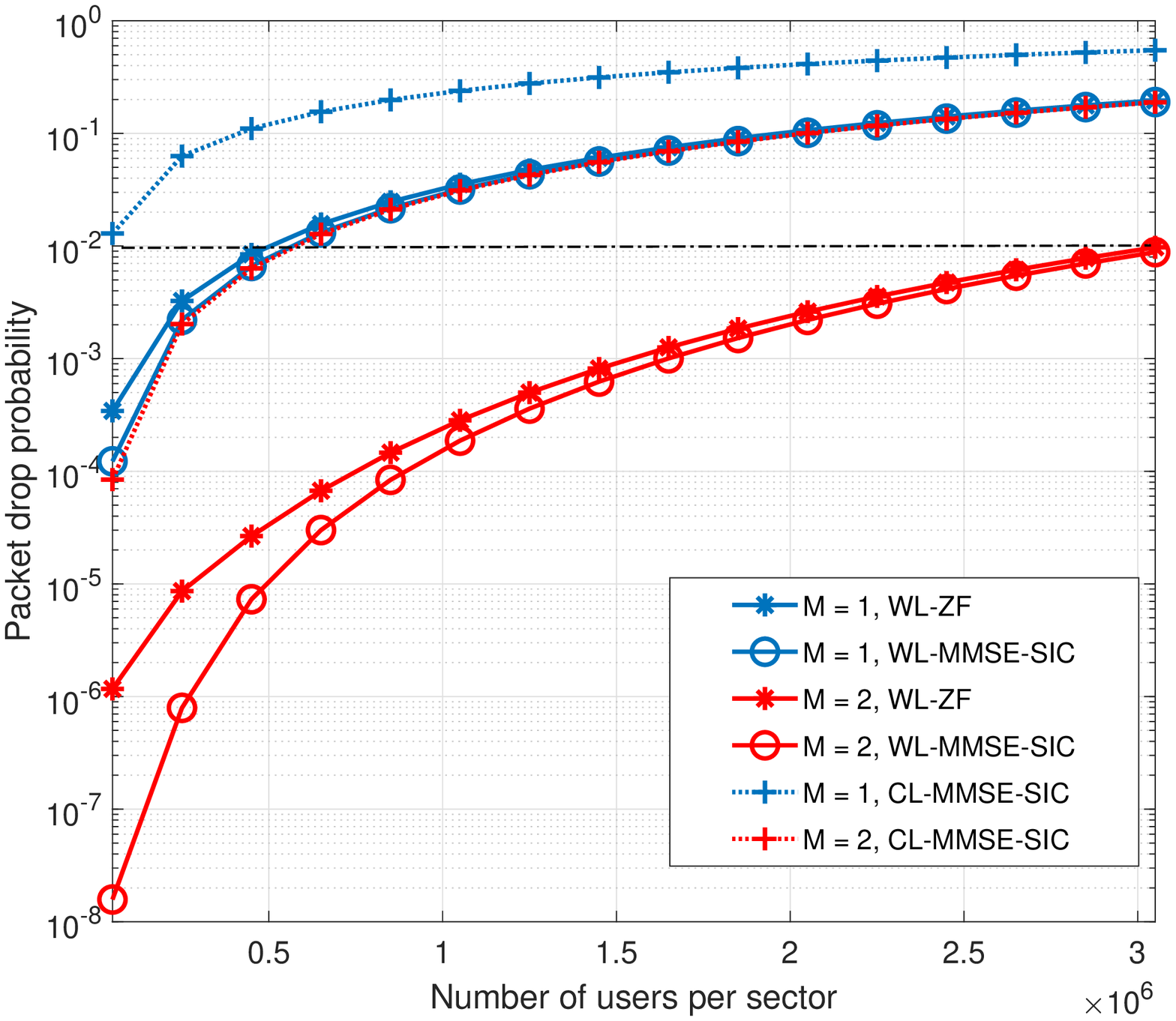}}
\subfigure[Throughput]{\includegraphics[width=0.5\textwidth]{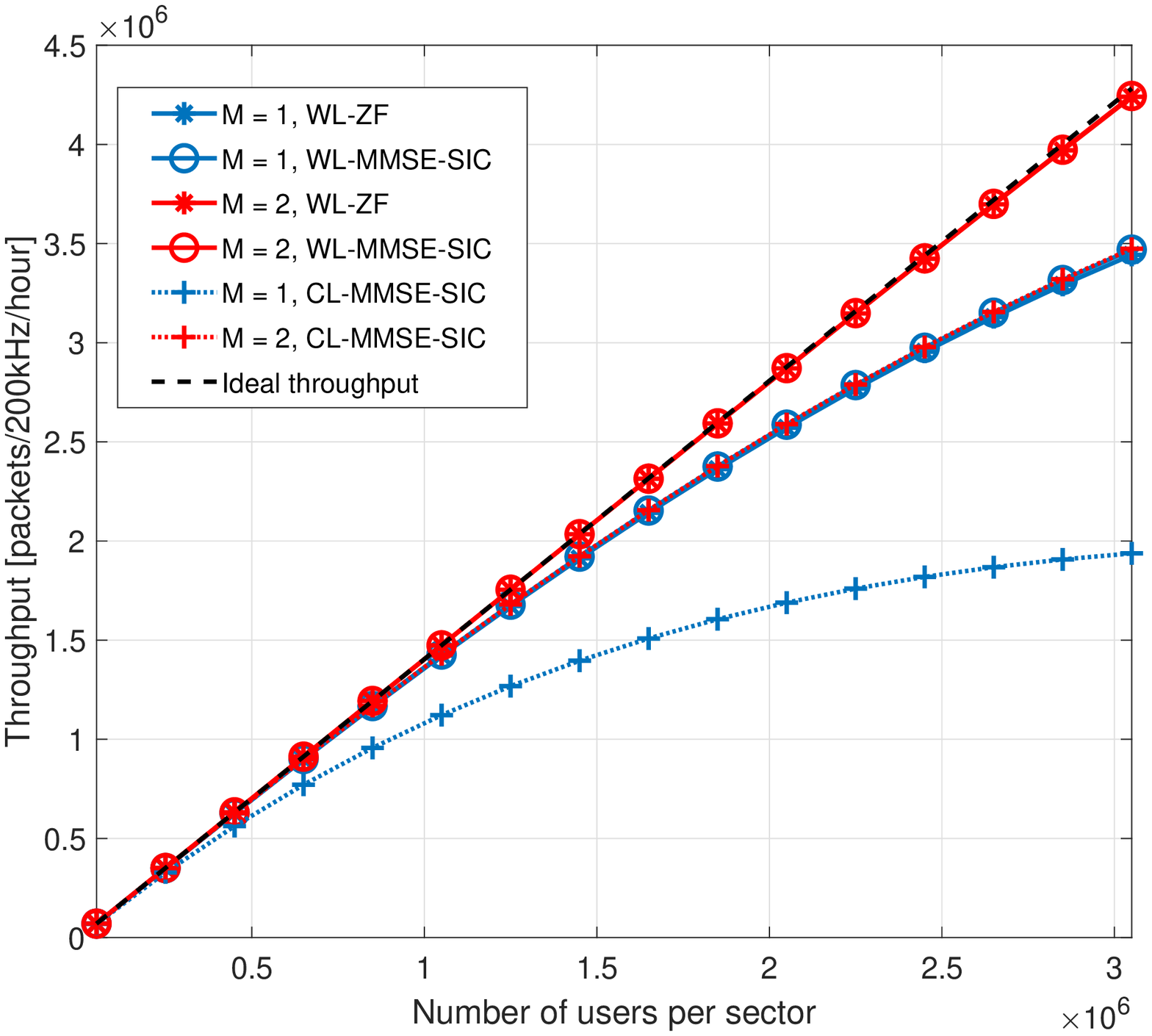}}
\caption{\label{fig:sameTTI}Performance comparison of CL and WL receivers with same TTI.}
\vspace{-0.5cm}
\end{figure*}

\subsubsection{Results}
Fig.~\ref{fig:sameTTI}(a) and Fig.~\ref{fig:sameTTI}(b) demonstrate the results for packet drop probability and system throughput, respectively, computed as per our previous work in \cite{GuiBalasubramanyaConnectivity19}. Specifically, we have considered the best performing receivers for CL and WL, CL-MMSE-SIC and WL-MMSE-SIC, respectively, and the simplest WL receiver (WL-ZF). It is evident from Fig.~~\ref{fig:sameTTI}(a) that the number of users supported by CL is around 50k users per cell, while that using WL is about 500k for $M = 1$. When $M =2$, CL and WL support around 650k  users and 3 million users per cell, respectively. The improvement obtained in WL processing is because of its ability to resolve collisions more effectively than CL. 

\begin{figure*}[t]
\subfigure[Packet drop probability]{\includegraphics[width=0.5\textwidth]{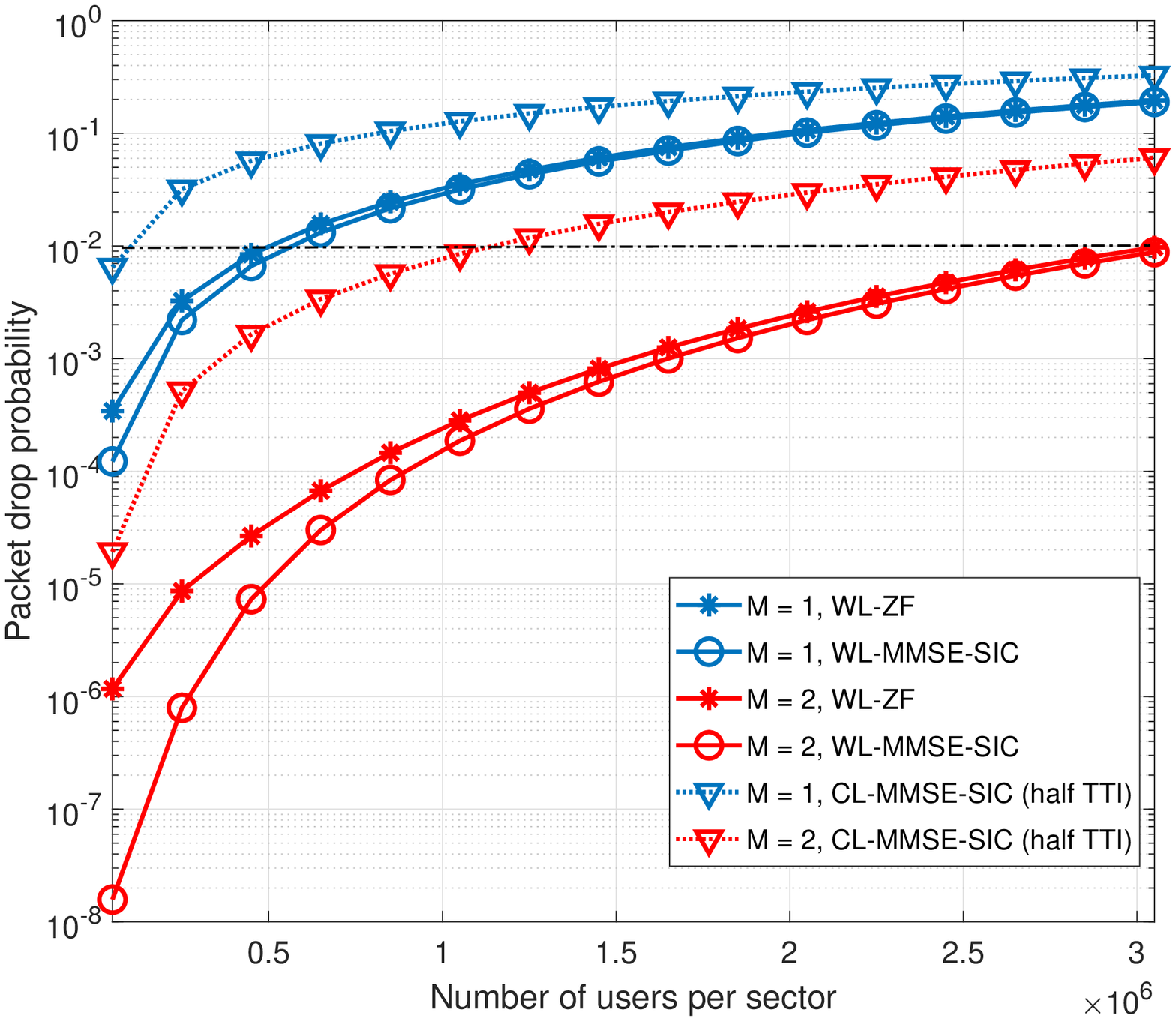}}
\subfigure[Throughput]{\includegraphics[width=0.5\textwidth]{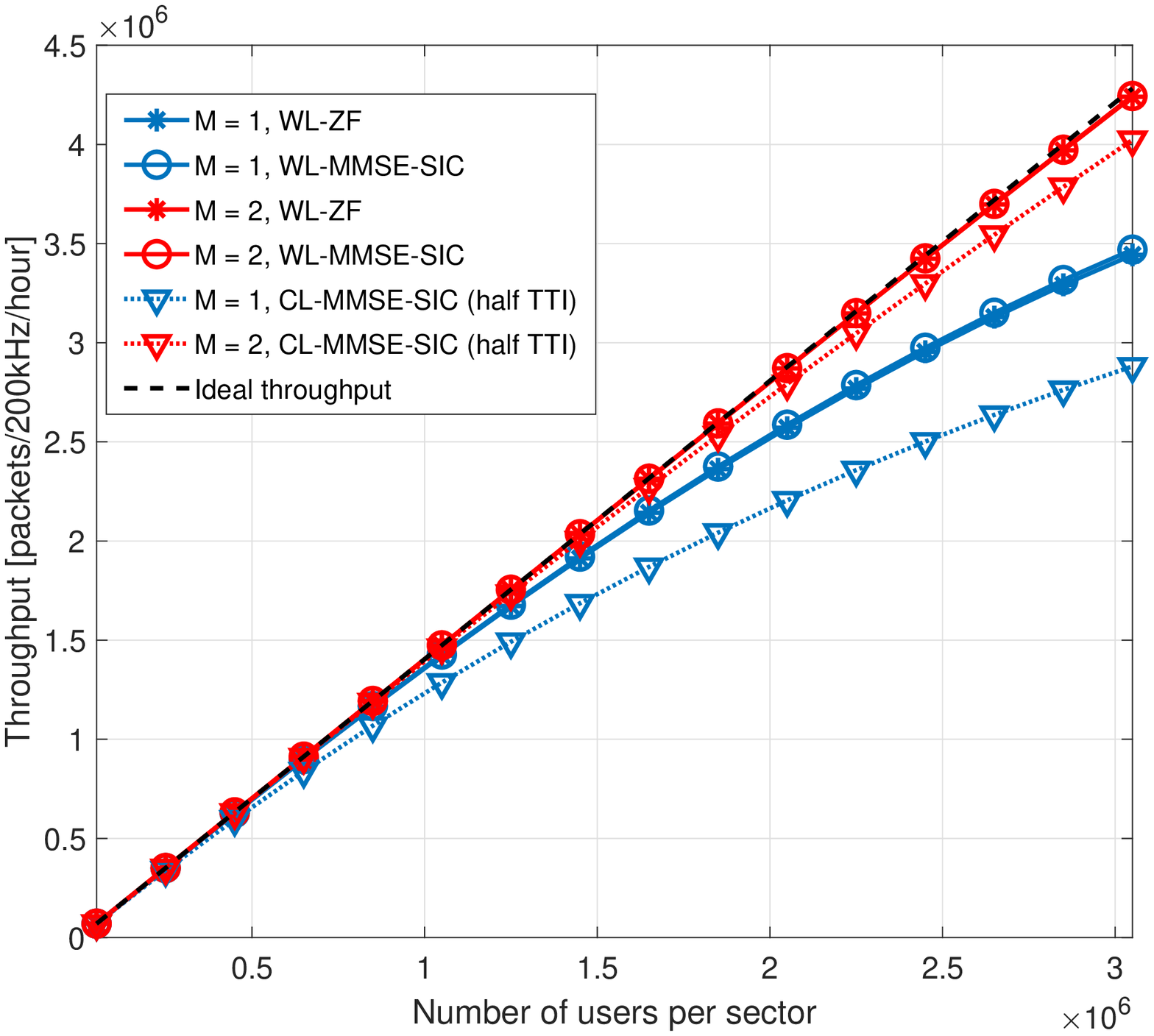}}
\caption{\label{fig:halfTTI}Performance comparison with CL and WL receivers with CL users occupying half TTI.}
\vspace{-0.5cm}
\end{figure*}

\subsubsection{\textit{Discussion}}
Hitherto we have considered that users transmit over the same duration. If this restriction is relaxed, the CL users could operate in a mode where the code rate is equal to that of the WL users, i.e., $r$, but the transmission rate is $2R$ over half the duration (TTI/2) as that of WL users (in order to maintain the same target rate $R$). Such a mode would be advantageous for CL users, since it helps to reduce collision rates owing to the shorter transmission interval. In this case, the outage probability analysis can be done by replacing $R$ with $2R$ in the equations corresponding to CL receivers. The packet-drop and throughput in this case can be analyzed by using the packet transmission rate as $\lambda_a/2$ per TTI.  The results corresponding to this mode of operation are indicated by the label CL (half TTI) in Fig.~\ref{fig:halfTTI}. As expected, CL (half TTI) results in lower packet drop probability when compared to CL using the same TTI. Consequently, for a packet drop probability of $1\%$, the CL (half TTI) supports around 80k users per cell and 1.2 million users per cell for $M = 1$ and $M = 2$, respectively.

We note that one could choose a larger constellation and reduce the transmission interval for WL users too, e.g. the combination of 4-PAM and TTI/2 for WL users. However, even without this alteration, WL processing supports higher number of users (500k and 3 million users per cell, respectively) than CL (half TTI), since it is better at resolving collisions. Furthermore, we would like to point that an optimization of the transmission duration would not be suitable for grant-free transmission.


Thus, our analysis and the results above indicate that the use of real-valued transmission with WL detection for uplink transmission would support a larger number of users, i.e., $N>M$ per time-frequency resource.

\section{Conclusions}
\label{s:conclusions}

In this paper, we have analyzed the outage performance of various WL receivers in an uplink multi-user MIMO system with real-valued modulations, in terms of diversity and coding gains. We prove that WL receivers with $N_\textrm{WL}$ real-valued user signals have a higher diversity gain than their CL counterparts with $N_\textrm{CL}$ complex-valued user signals when $N_\textrm{WL} < 2N_\textrm{CL} - 1$ and provide the same diversity gain when $N_\textrm{WL} = 2N_\textrm{CL} - 1$. We show that WL receivers with $M$ antennas can support a maximum number of $N_\textrm{WL} = 2M$ users with a positive diversity gain. Amongst the WL receivers serving $N$ users, we show that the WL-ZF and WL-MMSE receivers have the same diversity gain ($d_\mathrm{WL} = M -(N-1)/2$), but different coding gains. Moreover, the SIC operation with channel-dependent ordering brings no additional diversity gain to the WL-ZF and WL-MMSE receivers but increases their coding gains. The increase in coding gain due to SIC grows as the number of users increases. Furthermore, through analytical comparisons and numerical results, we show that the WL receivers are superior to the CL receivers. For grant-free mMTC (with a fixed TTI), WL processing results in improved transmission reliability (higher diversity gain) and user-multiplexing capability (supporting more users). Finally, we note that the WL receiver performance has been analyzed for Rayleigh fading channels. The analysis for more general cases, such as  Ricean fading channels, is expected to be more intricate and will thus be an interesting direction of work in the near future. We conclude that WL processing would be beneficial for numerous mMTC applications, such as, data transfer from smart metering systems, industrial asset tracking,  transport fleet management, smart buildings, etc. that demand low per-user data rates and are also delay tolerant while supporting high user densities.

	\appendices

	\section{Proof of Theorem~\ref{PolExpofCdf-k}}
        \label{p:theorem1}

We first present three Lemmas that we will use for proving Theorem~\ref{PolExpofCdf-k}.
	\begin{lemma}\label{DetIntegral}\cite{ChianiDistribution14}
		For an arbitrary $n \times n$ matrix $\ve{\Phi}(\ve{x})$ whose $(i,j)$th entry can be expressed as $\left[ \ve{\Phi}(\ve{x}) \right] _{i,j}=\Phi_i( x_j)$ for $1\le i,j\le n$, the following identity holds:
		\begin{equation}\label{Pff}
		\int_{\mathcal{D}_{\text{ord}}}{\left| \ve{\Phi }(\ve{x}) \right|\mathrm{d}\ve{x}}=\text{Pf}[\ve{A}], 
		\end{equation}
		where $\mathcal{D}_{\text{ord}}=\left\{ a\le x_1\le \ldots \le x_n\le b \right\}$ stands for the integral domain with ordered entries of $\ve{x}=[ x_1\ x_2\ \ldots \ x_n] ^T$, $\left| \ve{\Phi }(\ve{x}) \right|$ denotes the determinant of $\ve{\Phi }(\ve{x})$ and $\text{Pf}[\ve{A}]$ is the Pfaffian of the skew-symmetric matrix $\ve{A}$. The Pfaffian of skew-symmetric matrix $\ve{A}$ is a polynomial in the entries of $\ve{A}$ whose square is the determinant $|\ve{A}|$, i.e., $\left[ \text{Pf}(\ve{A}) \right] ^2=|\ve{A}|$.
			
		For $n$ even, matrix $\ve{A}$ is $n\times n$ with its $(i,j)$th entry
		\begin{equation}\label{aij}
		a_{i,j}=\int_a^b{\int_a^b{ \mathcal{R}(y-x) \Phi_i(x) \Phi_j(y) \mathrm{d}x\mathrm{d}y}},\quad 1\le i<j\le n.
		\end{equation}
		where $\mathcal{R}(y-x)$ is equal to $1$ for $y\ge x$ and to $-1$ otherwise. For $n$ odd, the skew-symmetric matrix $\ve{A}$ is $(n+1)\times (n+1)$ with $a_{i,j}$ given by (\ref{aij}) for $1\le i<j\le n$. The additional entries are $a _{i,n+1}=-a_{n+1,i}=\int_a^b{\Phi_i(x) \mathrm{d}x}$ for $1\le i \le n$ and $a_{n+1,n+1}=0$.
	\end{lemma}

\begin{lemma}\label{PffExpansion}\cite{OkadaPfaffian19}
	Let $\ve{A}\in \mathbb{R}^{2m\times 2m}$ be a skew-symmetric matrix with its $(i,j)$th entry $a_{i,j}$, $1\le i,j \le 2m$. The Pfaffian of $\ve{A}$ can be expressed as
	\begin{subequations}\label{Pf}
		\begin{align}\label{e:Pfa}
		\text{Pf}[ \ve{A}] &=\sum_{\sigma \in \mathcal{S}}{\sgn(\sigma) \prod_{i=1}^m{a_{\sigma _{2i-1},\sigma _{2i}}}}
		\\\label{e:Pfb}
		&=\sum_{k=1}^m{\sum_{\sigma \in \mathcal{S}_k}{\sgn(\sigma) a_{\sigma _{2k-1},2m}\prod_{i=1,i\ne k}^m{a_{\sigma _{2i-1},\sigma _{2i}}}}}
		\end{align}
	\end{subequations}
	where $\sigma =\sigma _1,\sigma _2,\cdots ,\sigma _{2m}$ is a permutation of the integers $1,2,...,2m$, ${\cal S}$ is the set of all permutations $\sigma$ satisfying $\sigma_1 <\sigma_3 <\cdots <\sigma_{2m-1}$ and $\sigma_{2i-1} <\sigma_{2i}$ for $1\le i \le m$, and ${\cal S}_{k}$ is the $k$th subset of ${\cal S}$ including all permutations satisfying $\sigma_{2k} =2m$. In (\ref{e:Pfa}) and (\ref{e:Pfb}), the sums are over all permutations in the set ${\cal S}$ and ${\cal S}_k$, respectively. $\sgn(\sigma)$ denotes the sign of permutation $\sigma$, which equals to either $+1$ or $-1$ and is determined by the number of transpositions to obtain the permutation.
\end{lemma}

\begin{lemma}\label{l:pdfeig}\cite{ZhongMcKayDistribution11}
Let $0\le \lambda _1\le \lambda_2\le \cdots \le \lambda _n<\infty$ be the ordered eigenvalues of the real-valued central Wishart matrix $\ve{XX}^T\in \mathbb{R}^{n\times n}$, where each entry of $\ve{X}$ is i.i.d.\ drawn from the standard normal distribution. The joint PDF of $\ve{\lambda }=[\lambda _1,\lambda _2,\cdots ,\lambda _n ]^T$ is expressed as
\begin{equation}\label{jpdf}
f_{\ve{\lambda }}(\ve{\phi }) =K_{nm}^{-1}\left| \ve{U}(\ve{\phi}) \right|\left[ \prod_{i=1}^n{\e^{-\frac{1}{2}\phi _i}\phi _{i}^{\frac{1}{2}( m-n-1)}} \right],
\end{equation}
where $\ve{\phi }=\left[ \phi _1\ \phi _2\ldots \phi _n \right] ^T$ with $0\le \phi _1\le \phi_2\le \ldots \le \phi _n<\infty$, $\ve{U}(\ve{\phi}) \in \mathbb{R}^{n\times n}$ is a Vandermonde matrix with the $(i,j)$th entry $\left[ \ve{U}(\ve{\phi }) \right] _{i,j}=\phi _{j}^{i-1}$ and $\left| \ve{U}(\ve{\phi }) \right|$ is its determinant, given by
\begin{equation}\label{VanDet}
\left| \ve{U}(\ve{\phi }) \right|=\prod_{i=1}^{n-1}{\prod_{j=i+1}^n{(\phi_j-\phi_i)}}.
\end{equation}
\end{lemma}

From Lemma~\ref{l:pdfeig}, we can express the marginal PDF of the $k$th ordered eigenvalue as
\begin{equation}\label{MpdfExp}
f_{\lambda _k}(\phi _k) =\int\limits_{\mathcal{D}_{k}^{+}}{\int\limits_{\mathcal{D}_{k}^{-}}{f_{\ve{\lambda }}(\ve{\phi }) \mathrm{d}\ve{\phi}_{k}^{-}\mathrm{d}\ve{\phi }_{k}^{+}}}
=\int\limits_{\mathcal{D}_{k}^{+}}{\int\limits_{\mathcal{D}_{k}^{-}}{K_{nm}^{-1}\left| \ve{U}(\ve{\phi}) \right|\left[ \prod_{i=1}^n{\text{e}^{-\frac{1}{2}\phi _i}\phi_{i}^{\frac{1}{2}\left( m-n-1 \right)}} \right] \mathrm{d}\ve{\phi }_{k}^{-}\mathrm{d}\ve{\phi }_{k}^{+}}}
\end{equation}
where $\mathcal{D}_{k}^{-}=\left\{ 0\le \phi _1\le \ldots \le \phi _{k-1}\le \phi _k \right\}$ and $\mathcal{D}_{k}^{+}=\left\{ \phi_k\le \phi_{k+1}\le \ldots \le \phi_n<\infty \right\}$ are the integral domains with ordered eigenvalues $\ve{\phi }_{k}^{-}= [\phi _1\ \phi_2\ldots \phi _{k-1}]^T$ and $\ve{\phi }_{k}^{+}=[\phi _{k+1}\ \phi_{k+1}\ldots\phi _n ]^T$, respectively. Now we can proceed with the proof of Theorem~\ref{PolExpofCdf-k} as follows.
\begin{enumerate}
\item[]Step 1:  Express the joint PDF $f_{\ve{\lambda }}(\ve{\phi })$ as a product of three parts such that terms only depending on $\ve{\phi}_k^{-}$ and $\ve{\phi }_{k}^{+}$ are separated.
\item[]Step 2: Calculate the integrals in $\ve{\phi }_{k}^{-}$ and $\ve{\phi }_{k}^{+}$ using Lemmas~\ref{DetIntegral} and~\ref{PffExpansion}.
\item[]Step 3: Simplify the result using the polynomial approximation around zero. 
\end{enumerate}

\paragraph{Step 1} We can write \eqref{VanDet} as 
	\begin{equation}\label{DetSpit}
	\begin{split}
	\left| \ve{U}(\ve{\phi}) \right|=\left| \ve{U}(\ve{\phi }_{k}^{-}) \right|\prod_{i=1}^{k-1}{\prod_{j=k}^n{\left( \phi _j-\phi _i \right)}}\times \left| \ve{U}(\ve{\phi }_{k}^{+}) \right|\prod_{j=k+1}^n{(\phi _j-\phi _k)}
	\end{split}
	\end{equation}
and accordingly rewrite the joint PDF in (\ref{jpdf}) as 
	\begin{equation}\label{jpdf2}
	\begin{split}
	f_{\ve{\lambda }}(\ve{\phi}) =&\left|\ve{U}(\ve{\phi }_{k}^{-})\right| \prod_{i=1}^{k-1}\e^{-\frac{1}{2}\phi _i}\phi _{i}^{\frac{1}{2}\left( m-n-1 \right)}\prod_{j=k}^n{(\phi _j-\phi _i)}
\\
& \times K_{nm}^{-1}\e^{-\frac{1}{2}\phi _k}\phi _{k}^{\frac{1}{2}\left( m-n-1 \right)}
	\left| \ve{U}(\ve{\phi }_{k}^{+}) \right|\prod_{i=k+1}^n{\e^{-\frac{1}{2}\phi _i}\phi _{i}^{\frac{1}{2}\left( m-n-1 \right)}(\phi _i-\phi _k)},
	\end{split}
	\end{equation}
for which the product terms in the second line only depend on $\phi_k$ and $\ve{\phi }_{k}^{+}$.

	\paragraph{Step 2} We first consider the integration of $f_{\ve{\lambda }}(\ve{\phi})$ with respect to $\ve{\phi }_{k}^{-}$. In order to apply Lemma~\ref{DetIntegral}, we define
	\begin{align}
	\begin{split}
	g(\phi_i,\phi_k,\ve{\phi}_k^+) &\triangleq \e^{-\frac{1}{2}\phi _i}\phi _{i}^{\frac{1}{2}\left( m-n-1 \right)}\prod_{j=k}^n{(\phi _j-\phi _i)}
	\end{split}
	\\
	\begin{split}\label{Klmd}
	h(\phi_k,\ve{\phi }_{k}^{+}) &\triangleq K_{nm}^{-1}\e^{-\frac{1}{2}\phi _k}\phi _{k}^{\frac{1}{2}\left( m-n-1 \right)}\left| \ve{U}(\ve{\phi }_{k}^{+}) \right|\left[ \prod_{i=k+1}^n{\e^{-\frac{1}{2}\phi _i}\phi_{i}^{\frac{1}{2}\left( m-n-1 \right)}(\phi_i-\phi_k)} \right] . 
	\end{split}
	\end{align}
	and then express the joint PDF as
	\begin{equation}
	f_{\ve{\lambda }}( \ve{\phi}) =h(\phi_k,\ve{\phi}_{k}^{+}) \left| \ve{U}(\ve{\phi}_{k}^{-}) \right|\prod_{i=1}^{k-1}g(\phi _i,\phi_k,\ve{\phi}_k^+).
	\end{equation}
	By identifying $\Phi_i(x_j)$ in Lemma~\ref{DetIntegral} with $x_{j}^{i-1}g(x_j,\phi_k,\ve{\phi}_k^+)$, we can calculate the joint PDF of $\lambda _k$ and $\ve{\lambda }_{k}^{+}= [\lambda _{k+1},...,\lambda _n ]^T$ as
	\begin{equation}\label{jpdfk1}
	\begin{split}
	f_{\lambda _k,\ve{\lambda }_{k}^{+}}( \phi_k,\ve{\phi }_{k}^{+}) &=\int_{\mathcal{D}_{k}^{-}}{f_{\ve{\lambda }}(\ve{\phi }) \mathrm{d}\ve{\phi }_{k}^{-}}
	= h( \phi_k,\ve{\phi }_{k}^{+})\int_{\mathcal{D}_{k}^{-}}{\left| \ve{U}(\ve{\phi }_{k}^{-}) \right|\prod_{i=1}^{k-1}{g(\phi_i ,\phi_k,\ve{\phi}_k^+) \mathrm{d}\ve{\phi}_{k}^{-}}}
	\\
	&= h(\phi _k,\ve{\phi }_{k}^{+}) \text{Pf}[ \ve{J}(\phi_k,\ve{\phi}_k^+)],
	\end{split}
	\end{equation}
	where the entries of the skew-symmetric matrix $\ve{J}(\phi_k,\ve{\phi}_k^+)$ are  
	\begin{equation}\label{Jkij}
	\begin{split}
	\left[ \ve{J}(\phi_k,\ve{\phi}_k^+) \right] _{i,j}&=\int_0^{\phi _k}{\int_0^{\phi _k}{\mathcal{R}(y-x) x^{i-1}g( x,\phi_k,\ve{\phi}_k^+) y^{j-1}g(y, \phi_k,\ve{\phi}_k^+) \mathrm{d}x\mathrm{d}y}}
	\\
	&=\int_0^{\phi _k}{\int_0^{\phi _k}{e_{i,j}(x,y) \prod_{s=k}^n{(\phi _s-x) ( \phi_s-y)}\mathrm{d}x\mathrm{d}y}}
	\end{split}
	\end{equation}
	for $1\le i<j \le k-1$ and
	\begin{equation}\label{Jkij2}
	e_{i,j}(x,y ) =\mathcal{R}(y-x) \e^{-\frac{1}{2}x}\e^{-\frac{1}{2}y}x^{\frac{1}{2}( m-n-1) +(i-1)}y^{\frac{1}{2}( m-n-1) +( j-1)}\;.
	\end{equation}
        For odd $k$,  $\ve{J}(\phi_k,\ve{\phi}_k^+)$ has size $(k-1)\times (k-1)$, and for even $k$, 
        $\ve{J}(\phi_k,\ve{\phi}_k^+)$ has size $k\times k$, and the additional entries are given by
	\begin{equation}
	\left[ \ve{J}(\phi_k,\ve{\phi}_k^+)\right]_{i,k}=-\left[ \ve{J}(\phi_k,\ve{\phi}_k^+) \right]_{k,i}=\int_0^{\phi _k}{x^{i-1}g(x,\phi_k,\ve{\phi}_k^+) {\mathrm d}x}
	\end{equation}
	for $i<k$ and $\left[ \ve{J}(\phi_k,\ve{\phi}_k^+)\right]_{k,k}=0$. 
	
	 Inserting \eqref{jpdfk1} into \eqref{MpdfExp} yields the marginal PDF of $\lambda_k$ as
	\begin{equation}\label{flmdk}
	f_{\lambda _k}(\phi _k) =\int_{\mathcal{D}_{k}^{+}}{f_{\lambda _k,\ve{\lambda }_{k}^{+}}(\phi _k,\ve{\phi }_{k}^{+}) \mathrm{d}\ve{\phi}_{k}^{+}}
=\int_{\mathcal{D}_{k}^{+}}{h(\phi _k,\ve{\phi }_{k}^{+}) \text{Pf}[ \ve{J}(\phi_k,\ve{\phi}_k^+)] {\mathrm{d}}\ve{\phi }_{k}^{+}},
	\end{equation}
        and we turn to the integration with respect to $\ve{\phi }_{k}^{+}$. 
	For this, it is necessary to expand $\text{Pf}[ \ve{J}_{k}^{-}(\phi_k,\ve{\phi}_k^+)]$ via Lemma~\ref{PffExpansion}. 

 For odd $k$, we apply 
 (\ref{e:Pfa}) 
 to \eqref{flmdk} and obtain
	\begin{equation}\label{flmdkExp}
	f_{\lambda _k}(\phi_k) =\sum_{\sigma \in \mathcal{S}}{\sgn(\sigma) \int_{\mathcal{D}_{k}^{+}}{h(\phi _k,\ve{\phi}_{k}^{+}) \prod_{i=1}^{\left( k-1 \right) /2}{\left[ \ve{J}(\phi_k,\ve{\phi}_k^+) \right]_{\sigma _{2i-1},\sigma _{2i}}}\mathrm{d}\ve{\phi}_{k}^{+}}}.
	\end{equation}
	 Next, to apply Lemma~\ref{DetIntegral}, we consider the expression for $\ve{J}(\phi_k,\ve{\phi}_k^+)$ in \eqref{Jkij} and define
 $\ve{x}\triangleq [x_1,\cdots ,x_{\left( k-1 \right) /2}]^T$, $\ve{y}\triangleq  [y_1,\cdots ,y_{\left( k-1 \right) /2} ]^T$ 
 and $\mathcal{D}_{\ve{xy}}\triangleq\left\{ \ve{0}\le \ve{x}\le \phi_k\ve{1},\ve{0}\le \ve{y}\le \phi_k\ve{1}\right\}$, where $\ve{0}$ and $\ve{1}$ are the all-zero and all-one vector of length $(k-1)/2$, respectively, and we obtain 
	\begin{align}
\nonumber
&\prod_{i=1}^{\left( k-1 \right) /2}{\left[\boldsymbol{J}(\phi _k,\boldsymbol{\phi }_{k}^{+}) \right]_{\sigma _{2i-1},\sigma_{2i}}}=\int_{\mathcal{D}_{\boldsymbol{xy}}}\prod_{i=1}^{(k-1)/2}{e_{\sigma_{2i-1},\sigma _{2i}}(x_i,y_i)}\prod_{s=k}^n(\phi_s-x_i) (\phi _s-y_i)\mathrm{d}\boldsymbol{x}\mathrm{d}\boldsymbol{y}
	\\
	&\qquad=\int\limits_{\mathcal{D}_{\boldsymbol{xy}}}{\prod_{i=1}^{( k-1) /2}{e_{\sigma_{2i-1},\sigma_{2i}}(x_i,y_i)(\phi_k-x_i)( \phi_k-y_i)}}
	 \prod_{i=1}^{(k-1) /2}{\prod_{s=k+1}^n{( \phi_s-x_i) ( \phi_s-y_i )}} \mathrm{d}\boldsymbol{x}\mathrm{d}\boldsymbol{y}.
	\label{multifoldInt}
\end{align}
Substituting \eqref{multifoldInt} into \eqref{flmdkExp} yields
	\begin{equation}\label{multifoldInt2}
	f_{\lambda _k}(\phi_k) =\sum_{\sigma \in \mathcal{S}}{\sgn(\sigma) \int_{\mathcal{D}_{\ve{xy}}}{\prod_{i=1}^{( k-1) /2}{e_{\sigma _{2i-1},\sigma _{2i}}( x_i,y_i)}( \phi_k-x_i) (\phi_k-y_i) p(\phi _k,x_i,y_i) \mathrm{d}\ve{x}\mathrm{d}\ve{y}}},
	\end{equation}
	where
	\begin{equation}\label{e:pfunc}
	\begin{split}
	p\left( \phi _k,x_i,y_i \right) &=\int_{\mathcal{D}_{k}^{+}}{h\left( \phi _k,\boldsymbol{\phi }_{k}^{+} \right) \prod_{i=1}^{\left( k-1 \right) /2}{\prod_{s=k+1}^n{\left( \phi _s-x_i \right) \left( \phi _s-y_i \right) d\boldsymbol{\phi }_{k}^{+}}}}
		\\
		&=K_{nm}^{-1}\e^{-\frac{1}{2}\phi _k}\phi_{k}^{\frac{1}{2}( m-n-1)}\int_{\mathcal{D}_{k}^{+}}{\left| \boldsymbol{U}( \boldsymbol{\phi }_{k}^{+}) \right|\prod_{s=k+1}^n{q(\phi_k,\phi_s,x_i,y_i) \mathrm{d}\boldsymbol{\phi }_{k}^{+}}}
	\end{split}
	\end{equation}
	with
	\begin{equation}
	q( \phi_k,\phi_s,x_i,y_i) =\e^{-\frac{1}{2}\phi _s}\phi _{s}^{\frac{1}{2}( m-n-1)}( \phi_s-\phi_k) \prod_{i=1}^{( k-1) /2}{( \phi _s-x_i) ( \phi_s-y_i)}.
	\end{equation}
	We can now apply Lemma~\ref{DetIntegral} to \eqref{e:pfunc} and then obtain
	\begin{equation}\label{GammaExp2}
          p(\phi_k,x_i,y_i) =K_{nm}^{-1}\e^{-\frac{1}{2}\phi _k}\phi_{k}^{\frac{1}{2}( m-n-1 )}\text{Pf}[ \ve{\tilde J}( \phi_k,x_i,y_i)] ,
	\end{equation}
	where the concrete expression of the elements of the skew-symmetric matrix $\ve{\tilde J}(\phi _k,x_i,y_i)$ is not required for subsequent derivations and thus omitted. Substituting \eqref{GammaExp2} into \eqref{multifoldInt2} and reverting back to single-variate integrals, we finally get 
	\begin{equation}\label{PdfLmdk}
	f_{\lambda _k}(\phi_k) =K_{nm}^{-1}\e^{-\frac{1}{2}\phi _k}\phi_{k}^{\frac{1}{2}( m-n-1 )}\sum_{\sigma \in \mathcal{S}}{\sgn(\sigma) \prod_{i=1}^{( k-1)/2}{z_{\sigma _{2i-1},\sigma _{2i}}(\phi _k)}},
	\end{equation}
	where
	\begin{equation}\label{z_sigma_even}
	z_{\sigma _{2i-1},\sigma _{2i}}(\phi _k) =\int_0^{\phi _k}{\int_0^{\phi _k}{e_{\sigma _{2i-1},\sigma _{2i}}(x,y)}}(\phi _k-x)(\phi_k-y) \text{Pf}[ \ve{\tilde J}(\phi_k,x,y)] \mathrm{d}x\mathrm{d}y.
	\end{equation}
	
 For even $k$, we apply (\ref{e:Pfb}) to \eqref{flmdk} and following the analogous steps as above, obtain
 	\begin{equation}\label{PdfEvenCase}
 	f_{\lambda _k}(\phi _k) =K_{nm}^{-1}\text{e}^{-\frac{1}{2}\phi _k}\phi _{k}^{\frac{1}{2}( m-n-1)}\sum_{\ell =1}^{k/2}{\sum_{\sigma \in \mathcal{S}_{\ell}}{\sgn( \sigma) \tilde{z}_{\sigma _{2\ell-1},k}( \phi _k) \prod_{i=\text{1,}i\ne \ell}^{k/2}{z_{\sigma _{2i-1},\sigma _{2i}}( \phi _k)}}},
 	\end{equation}
 	where
 	\begin{equation}\label{zInt}
	\tilde{z}_{\sigma _{2\ell-1},k}( \phi _k) =\int_0^{\phi _k}{\e^{-\frac{1}{2}t}t^{\frac{1}{2}( m-n-1) +\sigma _{2\ell-1}-1}( \phi _k-t) \mathrm{d}t}.
 	\end{equation}

	\paragraph{Step 3} We consider the polynomial approximation of $f_{\lambda _k}(\phi_k)$ in \eqref{PdfLmdk} and \eqref{PdfEvenCase} around $\phi_k=0$.  
We write the first-order Taylor-series expansion of $\text{Pf}[\boldsymbol{\tilde{J}}( \phi _k,x,y)]$ as
	\begin{equation}\label{TayExp}
	\text{Pf}\left[ \boldsymbol{\tilde{J}}\left( \phi _k,x,y \right) \right] =\text{Pf}\left[ \boldsymbol{\tilde{J}}\left( 0,0, 0\right) \right] +D\left( \phi _k,x,y \right) 
	\end{equation}
	where $D\left( \phi _k,x,y \right) \triangleq D_{\phi}\phi _k+D_xx+D_yy+o\left( \phi _k \right) +o\left( x \right) +o\left( y \right)$, with $D_{\phi}$, $D_x$ and $D_y$ being the partial derivatives of $\text{Pf}[\boldsymbol{\tilde{J}}( \phi _k,x,y )]$ with respect to $\phi _k$, $x$ and $y$, all evaluated at $0$. 
	As $\phi_k\rightarrow 0$, it is not hard to show that
	\begin{equation}\label{e:helper}
	\int_0^{\phi _k}{\int_0^{\phi _k}{e_{i,j}( x,y) ( \phi _k-x)( \phi _k-y) {\mathrm{d}}x{\mathrm{d}}y}}=\alpha_{i,j}\phi_{k}^{( m-n+1) +i+j}+o\left( \phi _{k}^{( m-n+1) +i+j} \right),
	\end{equation}
	where $\alpha _{i,j}$ is a coefficient determined by $e_{i,j}(x,y)$. 
Consequently, we have
	\begin{equation}\label{PolExpPi}
	z_{\sigma _{2i-1},\sigma _{2i}}(\phi _k) =\text{Pf}[ \boldsymbol{\tilde{J}}( 0,0,0 )]\alpha _{\sigma_{2i-1},\sigma_{2i}}\phi _{k}^{( m-n+1) +\sigma_{2i-1}+\sigma_{2i}}+o\left( \phi_{k}^{( m-n+1) +\sigma_{2i-1}+\sigma_{2i}} \right). 
	\end{equation}

	Furthermore, the first-order expansion of $\tilde{z}_{\sigma _{2i-1},\sigma _k}\left( \phi _k \right)$ around $\phi _k=0$ is given by
	\begin{equation}\label{zExpEvenCase}
	\tilde{z}_{\sigma _{2\ell-1},k}\left( \phi _k \right) =\tilde{\alpha}_{\sigma _{2\ell-1},k}\phi _{k}^{\frac{1}{2}\left( m-n-1 \right) +\sigma _{2\ell-1}+1}+o\left( \phi _{k}^{\frac{1}{2}\left( m-n-1 \right) +\sigma _{2\ell-1}+1} \right) 
	\end{equation}
	where $\tilde{\alpha}_{\sigma _{2\ell-1},k}$ is determined by the integrand in \eqref{zInt}.

     Then, inserting \eqref{PolExpPi} into \eqref{PdfLmdk} yields the polynomial approximation of the marginal PDF of $\lambda_k$ for odd $k$, and using \eqref{PolExpPi} and \eqref{zExpEvenCase} in \eqref{PdfEvenCase} yields the polynomial approximation of the marginal PDF of $\lambda_k$ for even $k$.

	Next, noting that $\e^{-\frac{1}{2}\phi _k}= 1+o\left(1\right)$ and 
	\begin{subequations}
	\begin{align}
	\prod_{i=1}^{( k-1) /2}{\phi _{k}^{( m-n+1) +\sigma _{2i-1}+\sigma _{2i}}}=\phi_{k}^{\frac{1}{2}( k-1)( m-n+1) +\frac{1}{2}k( k-1)}
	\\
	\phi _{k}^{\frac{1}{2}( m-n-1) +\sigma _{2\ell -1}+1}\prod_{i=\text{1,}i\ne \ell}^{k/2}{\phi_{k}^{( m-n+1) +\sigma _{2i-1}+\sigma _{2i}}}=\phi _{k}^{\frac{1}{2}( k-1 ) ( m-n+1) +\frac{1}{2}k( k-1)}
	\end{align}
	\end{subequations}
	we obtain
	\begin{equation}
	f_{\lambda _k}( \phi _k) =K_{nm}^{-1}\text{Pf}[ \boldsymbol{J}_{\alpha}] \phi_{k}^{\tilde{d}_k}+o( \phi_{k}^{\tilde{d}_k}) ,
	\end{equation}
	where, for odd $k$, $\ve{J}_{\alpha}$ is a $(k-1) \times (k-1)$ skew-symmetric matrix with its $(i,j)$th entry given by $\left[ \ve{J}_{\alpha} \right] _{i,j}=\text{Pf}[ \boldsymbol{\tilde{J}}( 0,0,0 )]\alpha _{i,j}$, and for even $k$, $\ve{J}_{\alpha}$ is $k \times k$, with the additional entries $[ \boldsymbol{J}_{\alpha}] _{i,k}=-[ \boldsymbol{J}_{\alpha}]_{k,i}=\tilde{\alpha}_{i,k}$ for $1\le i \le k-1$ and $[ \ve{J}_{\alpha} ] _{k,k}=0$, and $\tilde{d}_k=\frac{1}{2}k\left( m-n+k-2 \right) +\left( k-1 \right)$. 

	Hence, the marginal CDF of $\lambda_k$ is polynomially expanded as
	\begin{equation}
	F_{\lambda _k}(\phi_k) =\beta _k\phi_{k}^{d_k}+o( \phi_{k}^{{d}_k}),
	\end{equation}
	where
	\begin{subequations}
	\begin{align}
	\beta_k&=d_{k}^{-1}K_{nm}^{-1}\text{Pf}[\boldsymbol{J}_{\alpha} ] 
	\\
	d_k&=\frac{1}{2}k( m-n+k).
	\end{align}
	\end{subequations}

	\section{Proof of Theorem~\ref{PolExpofCdf-1}}
\label{p:theorem2}

	Similar to the proof of Theorem~\ref{PolExpofCdf-k}, this proof makes use of Lemma~\ref{DetIntegral} and performs the three steps of  expressing the joint PDF $f_{\ve{\lambda }}(\ve{\phi })$ as a product of terms (Step~1), marginalizing the PDF (Step 2), and developing the polynomial expansion (Step 3). A final 4th step is added to obtain the coefficients in \eqref{Jij}.
	
	\paragraph{Step 1} Starting from \eqref{jpdf2} for $k=1$, and defining 
        \begin{equation}
	g( \phi_1,\phi_i) \triangleq \e^{-\frac{1}{2}\phi_i}\phi_{i}^{\frac{1}{2}\left( m-n-1 \right)}( \phi_i-\phi_1).
	\end{equation}
     the joint PDF $f_{\ve{\lambda }}( \ve{\phi } )$ is
	\begin{equation}
	f_{\ve{\lambda }}(\ve{\phi }) =K_{nm}^{-1}\e^{-\frac{1}{2}\phi_1}\phi _{1}^{\frac{1}{2}( m-n-1)}\left|\ve{U}(\ve{\phi }_{1}^{+}) \right|\prod_{i=2}^n{g(\phi_{1}, \phi_i )}.
	\end{equation}
	
	\paragraph{Step~2} The marginal PDF of $\lambda_1$ is calculated as
	\begin{equation}\label{fpdf1}
	f_{\lambda_1}( \phi_1) = K_{nm}^{-1}\e^{-\frac{1}{2}\phi _1}\phi _{1}^{\frac{1}{2}( m-n-1)}\int_{\mathcal{D}_{1}^{+}}{\ve{U}(\ve{\phi }_{1}^{+}) \prod_{i=2}^n{g(\phi_{1},\phi _i)}\mathrm{d}\ve{\phi }_{1}^{+}}.
	\end{equation}
        	By identifying $\Phi_i(x_j)$ in Lemma~\ref{DetIntegral} with $x_{j}^{i-1}g(\phi_1,x_j)$ in \eqref{fpdf1}, we can solve the marginal PDF as 
	\begin{equation}
	f_{\lambda_1}( \phi_1) =K_{nm}^{-1}\e^{-\frac{1}{2}\phi _1}\phi_{1}^{\frac{1}{2}( m-n-1)}\text{Pf}\left[ \ve{J}( \phi_1) \right],
	\end{equation}
	where $\ve{J}(\phi_1)$ is a skew-symmetric matrix of size $(n-1)\times (n-1)$ for odd $n$ and of size $n \times n$ for even $n$. For odd $n$, the elements of $\ve{J}( \phi_1 )$ are
	\begin{equation}
	\left[\ve{J}( \phi_1)\right]_{i,j} =\int_{\phi_1}^{\infty}{\int_{\phi_1}^{\infty}{\mathcal{R}(y-x) x^{i-1}g(\phi_1,x) y^{j-1}g(\phi_1,y) \mathrm{d}x\mathrm{d}y}},
	\end{equation}
	for $1\le i<j \le n-1$. For $n$ even, the additional entries of $\ve{J}(\phi _1)$ are
	\begin{equation}
	\left[\ve{J}( \phi_1)\right]_{i,n} =-\left[\ve{J}(\phi_1)\right]_{n,i}=\int_{\phi_1}^{\infty}{\e^{-\frac{1}{2}x}x^{\frac{1}{2}\left( m-n-1 \right) +\left( i-1 \right)}\left( x-\phi _1 \right) \mathrm{d}x}
	\end{equation}
	for $1 \le i \le n-1$ and $\left[\ve{J}(\phi_1)\right]_{n,n} =0$.
	
	\paragraph{Step~3} Noting that $\e^{-\frac{1}{2}\phi_1}=1+o(1) $ and
	$\text{Pf}\left[ \ve{J}(\phi_1) \right] =\text{Pf}\left[ \ve{J}( 0) \right] +o(1)$
	 for $\phi_1\to 0$, we have the polynomial expansion of $f_{\lambda_1}(\phi_1)$ around $\phi_1=0$ as
	\begin{equation}
	f_{\lambda_1}(\phi_1) =K_{nm}^{-1}\sqrt{|\ve{J}|}\phi_{1}^{\frac{1}{2}( m-n-1)}+o\left(\phi_{1}^{\frac{1}{2}( m-n-1)}\right) ,
	\end{equation}
	where we used  $\sqrt{|\ve{J}(0)|}=\text{Pf}\left[\ve{J}(0)\right]$ and introduced the short-hand notation $\ve{J}\triangleq\ve{J}(0)$ also used in Theorem~\ref{PolExpofCdf-1} and \eqref{Jij}. The CDF of $\lambda_1$ is readily obtained as
	\begin{equation}
	F_{\lambda_1}(\phi_1) =\beta_1\phi_{1}^{d_1}+o(\phi_{1}^{d_1}),
	\end{equation}
	with $d_1=\frac{1}{2}(m-n+1)$ as in \eqref{e:d1def} and $\beta_1=K_{nm}^{-1}d_{1}^{-1}\sqrt{|\ve{J}|} $as in \eqref{e:betadef}.

        \paragraph{Step 4} The final step is to calculate the entries of $\ve{J}$. For this, let  $b_i=\frac{1}{2}\left( m-n+1 \right) +i$ and 
$\Gamma_{\text{low}}( a,x) = \int_0^x\e^{-t}t^{a-1}\mathrm{d}t$
be the lower incomplete  Gamma function. Then, for $1\le i<j\le n-1$
	\begin{equation}
	\begin{split}
	[\ve{J}]_{i,j} &=\int_0^{\infty} \int_0^{\infty} \mathcal{R}( y-x) \e^{-\frac{1}{2}x}x^{b_i-1} \e^{-\frac{1}{2}y}y^{b_j-1}\mathrm{d}x\mathrm{d}y	\\
	&=\int_0^{\infty}\int_0^y \e^{-\frac{1}{2}x}x^{b_i-1} \e^{-\frac{1}{2}y}y^{b_j-1}\mathrm{d}x\mathrm{d}y-\int_0^{\infty}\int_y^{\infty}\e^{-\frac{1}{2}x}x^{b_i-1} \e^{-\frac{1}{2}y}y^{b_j-1}\mathrm{d}x\mathrm{d}y\\
&=\int_0^{\infty}2^{b_i}\Gamma_{\text{low}}(b_i,y/2) \e^{-\frac{1}{2}y}y^{b_j-1}\mathrm{d}y
-\int_0^{\infty} 2^{b_i}\left(\Gamma(b_i)-\Gamma_{\text{low}}(b_i,y/2)\right) \e^{-\frac{1}{2}y}y^{b_j-1}\mathrm{d}y.
	\end{split}
	\end{equation}
Defining
\begin{equation}
	I\left( a,b;x \right) \triangleq \frac{2}{\Gamma(a)\Gamma(b)}\int_0^{x}{t^{a-1}\e^{-t}\Gamma_{\text{low}}(b,t ) \mathrm{d}t}
	\end{equation}
we obtain the compact expression
	\begin{equation}
	[\ve{J}]_{i,j} = 2^{b_i}2^{b_j}\Gamma( b_i) \Gamma( b_j) \left[ 2I( b_j,b_i,\infty) -1 \right] \;.
	\end{equation}
From  \cite[Eqs. (15)-(17)]{ChianiDistribution14}, we have 
	\begin{equation}\label{RecursiveResult}
	I\left( b_j,b_i;\infty \right) =\frac{1}{2}-\sum_{k=1}^{i-j}{\frac{2^{-\left( b_j+b_i-k \right)}\Gamma \left( b_j+b_i-k \right)}{\Gamma \left( b_j \right) \Gamma \left( b_{i-k+1} \right)}},\ i\ge j.
	\end{equation}
	Hence, we can calculate $[\ve{J}]_{j,i}$ for $1\le i \le j \le n-1$ and then obtain $[\ve{J}]_{i,j} = -[\ve{J}]_{j,i}$, as in \eqref{Jij}. For $1\le i\le n-1$, we can readily obtain
	\begin{equation}
	[ \boldsymbol{J}] _{i,n}=-[ \boldsymbol{J}]_{n,i}=\int_0^{\infty}{\e^{-\frac{1}{2}x}x^{\frac{1}{2}( m-n-1 ) +i}\mathrm{d}x}=2^{b_i}\Gamma( b_i).
	\end{equation}

\section{Proof of (\ref{Moments})}\label{App-C}

For \eqref{e:momentcl}, recall that $\mu_n=|\nu_{n,1}|^2$, $\mu_{\min}=\min_n\{\mu_n\}$, $\nu_{n,1}$ is the first entry of $\ve{\nu}_n$ and $\ve{\nu}_n$ is a unit-length eigenvector of the complex central Wishart matrix $\boldsymbol{\bar{H}}^H\boldsymbol{\bar{H}}\in \mathbb{C}^{N_{\text{CL}}\times N_{\text{CL}}}$, where $\boldsymbol{\bar{H}}\in \mathbb{C}^{M\times N_{\text{CL}}}$ is the complex channel matrix. It has been shown in \cite{JiangVaranasiPerformance11} that $N_{\text{CL}}\mu _{\min}\sim \mu_n$ holds for all $N_{\text{CL}}$.  Hence, we arrive at $\mathbb{E}\left\{ \left( N_{\text{CL}}\mu _{\min} \right) ^d \right\} =\mathbb{E}\left\{ \mu _{n}^{d} \right\}$, or equivalently \eqref{e:momentcl}. To prove \eqref{e:momentwl}, we use the following lemma.
\begin{lemma}\label{L:CDF}
Let $\alpha _1$, $\alpha _2$, ..., $\alpha_N$ be i.i.d.\ standard Gaussian random variables and define $\zeta _i\triangleq \alpha _{i}^{2}$, $1\le i\le N$ and $\zeta _{\min}\triangleq \min_{1\le i\le N}\left\{ \zeta _i \right\}$. Then, given an arbitrarily small $\epsilon _0$, there is a sufficiently large $N$ such that 
\begin{equation}\label{L:momentIneq}
\frac{\mathbb{E}\{\zeta _{i}^{d}\}}{\mathbb{E}\{\left( N\zeta _{\min} \right)^d\}}\ge \frac{1}{1+\epsilon_0}.
\end{equation}
\end{lemma}

\begin{proof}
Since $\zeta_i$, $1\le i\le N$, are i.i.d.\ Chi-squared random variables with one DoF, we obtain the CDFs of $\zeta_i$ and $N\zeta_{\min}$ as
\begin{subequations}
\begin{align}
F_{\zeta _i}(x) &=\mathrm{erf}\left( \sqrt{\frac{x}{2}} \right) 
\\
F_{N\zeta _{\min}}(x) &=1-\left[ 1-\mathrm{erf}\left( \sqrt{\frac{x}{2N}} \right) \right]^N,
\end{align}
\end{subequations}
where $\mathrm{erf}(x)$ is the Gaussian error function. 

First, we express the desired moments using \cite{doi:10.1080/00031305.2017.1356374}
\begin{equation}
\label{e:expectation}
\mathbb{E}\{X^r\}=\int\limits_{0}^\infty  r x^{r-1}[1-F_X(x)] \mathrm{d}x\;\mbox{ for non-negative random variable }X,\\
\end{equation}
as
\begin{subequations}\label{e:moment}
\begin{align}
\mathbb{E}\{\zeta _{i}^{d}\}&=d\int\limits_0^{\infty}{x^{d-1}\left[ 1-\mathrm{erf}\left( \sqrt{x/2} \right) \right] \mathrm{d}x}=d\int\limits_0^{\infty}{x^{d-1}2Q\left( \sqrt{x} \right) \mathrm{d}x}
\\
\mathbb{E}\{\left( N\zeta _{\min} \right) ^d\}&=d\int\limits_0^{\infty}{x^{d-1}\left[ 1-\mathrm{erf}\left( \sqrt{x/\left( 2N \right)} \right) \right] ^N\mathrm{d}x}=d\int\limits_0^{\infty}{x^{d-1}\left[ 2Q\left( \sqrt{x/N} \right) \right] ^N\mathrm{d}x}
\end{align}
\end{subequations}
where $Q$ 
is the Gaussian Q-function. Noting 
$Q( x) \le \frac{1}{2}e^{-x^2/2}$,
we have
\begin{equation}
\mathbb{E}\{\left( N\zeta _{\min} \right) ^d\}\le d\int\limits_0^K{x^{d-1}\left[ 2Q\left( \sqrt{x/N} \right) \right] ^N\mathrm{d}x}+\epsilon \left( K \right)
\end{equation}
where
\begin{equation}
\epsilon \left( K \right) \triangleq d\int\limits_K^{\infty}{x^de^{-x/2}\mathrm{d}x}.
\end{equation}
Defining $\epsilon_0\triangleq\epsilon(K_0)$ for some $K_0$, it follows that
\begin{equation}
\frac{\mathbb{E}\{\zeta _{i}^{d}\}}{\mathbb{E}\{\left( N\zeta _{\min} \right) ^d\}}\ge \frac{\mathbb{E}\{\zeta _{i}^{d}\}}{d\int\limits_0^{K_0}{x^{d-1}\left[ 2Q\left( \sqrt{x/N} \right) \right] ^N\mathrm{d}x}+\epsilon_0}
 \ge \frac{d\int\limits_0^{K_0}{x^{d-1}2Q\left( \sqrt{x} \right) \mathrm{d}x}}{d\int\limits_0^{K_0}{x^{d-1}\left[ 2Q\left( \sqrt{x/N} \right) \right] ^N\mathrm{d}x}+\epsilon _0}. \label{MomentIneq}
\end{equation}
%

On the other hand, using the log-quadratic bounds for the Q-function \cite{mastin2013}
$Q\left( x \right) \le \frac{1}{2}\e^{-\frac{x^2}{\pi}-\sqrt{\frac{2}{\pi}}x}$ and $
Q\left( x \right) \ge \frac{1}{2}\e^{-\frac{x^2}{2}-\sqrt{\frac{2}{\pi}}x}$
%
we obtain the inequality
\begin{equation}\label{QfuncIneq}
\left[ 2Q\left( \sqrt{x/N} \right) \right] ^N
\le 2Q\left( \sqrt{x} \right) 
\end{equation}
for $N\ge N'(x)$ with 
\begin{equation}
N'\left( x \right) =\left[ 1+\left( \sqrt{\frac{\pi}{4}}-\sqrt{\frac{1}{2\pi}} \right) \sqrt{x} \right] ^2.
\end{equation}
Applying \eqref{QfuncIneq} in \eqref{MomentIneq}, we obtain, for $N\ge N'(K_0)$,
\begin{equation}
\frac{\mathbb{E}\{\zeta _{i}^{d}\}}{\mathbb{E}\{\left( N\zeta _{\min} \right) ^d\}}\ge \left(\frac{d\int\limits_0^{K_0}{x^{d-1}\left[ 2Q\left( \sqrt{x/N} \right) \right] ^N\mathrm{d}x}}{d\int\limits_0^{K_0}{x^{d-1}2Q\left( \sqrt{x} \right) \mathrm{d}x}}+\epsilon _0 \right)^{-1}\ge \frac{1}{1+\epsilon _0}.
\end{equation}

In summary, since $\epsilon (K)$ monotonically decreases to $0$ as $K\rightarrow \infty$, given an arbitrarily small $\epsilon_0$, we can find a $K_0=\epsilon^{-1}( \epsilon _0)$, such that for $N\ge N'( K_0)$,
\begin{equation}
\frac{\mathbb{E}\{\zeta _{i}^{d}\}}{\mathbb{E}\{\left( N\zeta _{\min} \right) ^d\}}\ge \frac{1}{1+\epsilon _0}.
\end{equation}
\end{proof}

Recall that $u_n=|v_{n,1}|^2$, $u_{\min}=\min_n\{u_n\}$, and $v_{n,1}$ is the first entry of $\ve{v}_n$, which is the $n^{\mathrm{th}}$ column of $\ve{V}^T$. Thus the vector $\ve{v}\triangleq[v_{1,1},v_{2,1},\ldots,v_{N,1}]^T$ is the first unit-length eigenvector of $2\ve{H}^T\ve{H}\in \mathbb{R}^{N_{\text{WL}}\times N_{\text{WL}}}$ (see \eqref{omega}). From \cite[Lemma 2.1]{jiang:2010}, the eigenvector $\ve{v}$ can be generated according to {$\ve{v}\sim \frac{\ve{\alpha }}{\lVert \ve{\alpha } \rVert}$, where $\ve{\alpha }\sim \mathcal{N}\left( \ve{0},\ve{I}_{N_{\text{WL}}} \right)$. Hence, for asymptotically large $N_{\text{WL}}$, it follows that $\ve{v}\sim \mathcal{N}\left( \ve{0} ,\ve{I}_{N_{\text{WL}}} /N_{\text{WL}}\right)$. 
Then, we can apply Lemma~\ref{L:CDF} with $\epsilon_0\to 0$ for $N_{\text{WL}}\to\infty$ to arrive at \eqref{e:momentwl}.

	{
	\bibliographystyle{IEEEtran}
	\bibliography{IEEEabrv,ComRef_LL2}
	}	

\end{document}